\documentclass[11pt,a4paper]{article}
\usepackage{amsthm}
\usepackage{amsmath}
\usepackage{amssymb}
\usepackage{microtype}
\usepackage{mathtools}
\usepackage{a4wide}
\usepackage{color}
\usepackage{xspace}
\usepackage{hyperref}
\hypersetup{colorlinks=true,citecolor=blue,linkcolor=blue,urlcolor=blue}
\usepackage[capitalise]{cleveref} 
\usepackage{crossreftools}
\usepackage{tikz}
\usetikzlibrary{backgrounds,calc,decorations.pathreplacing,decorations}
\tikzset{
	v/.style={circle,fill=blue!70!yellow!50!gray!20, draw=black!75,inner sep=0pt,minimum size=5pt}
}

\newtheorem{theorem}{Theorem}[section]
\newtheorem*{theorem*}{Theorem}
\newtheorem{lemma}[theorem]{Lemma}
\newtheorem*{lemma*}{Lemma}
\newtheorem{corollary}[theorem]{Corollary}
\newtheorem{proposition}[theorem]{Proposition}
\newtheorem*{proposition*}{Proposition}

\theoremstyle{definition}
\newtheorem{definition}[theorem]{Definition}
\newtheorem{example}[theorem]{Example}
\newtheorem{remark}[theorem]{Remark}
\newtheorem{observation}[theorem]{Observation}

\Crefname{observation}{Observation}{Observations}
\crefname{equation}{}{}

\def\A{\mathbb{A}}
\def\B{\mathbb{B}}
\def\C{\mathbb{C}}
\def\D{\mathbb{D}}
\def\bbK{\mathbb{K}}
\def\bbI{\mathbb{I}}
\newcommand{\fA}{f^\A}
\newcommand{\fB}{f^\B}
\newcommand{\fC}{f^\C}
\newcommand{\fD}{f^\D}
\newcommand{\uA}{u^\A}

\newcommand{\uC}{u^\C}
\newcommand{\cA}{\mathcal{A}} % class of structures
\newcommand{\cP}{\mathcal{P}} % class of structures with planar Gaifman graphs
\newcommand{\cG}{\mathcal{G}} % class of Gaifman graphs

\newcommand\overcasts{\mathrel{\succeq}}
\newcommand\overcastsL{\mathrel{\preceq}}

\newcommand{\poleq}[1][_\bot]{\sqsubseteq_{#1}}
\newcommand{\pogeq}[1][_\bot]{\sqsupseteq_{#1}}
\def\circbot{\circ_{\hspace*{-2.5pt}\scriptscriptstyle\bot\hspace*{-1pt}}}

\DeclareMathOperator{\dist}{d}
\newcommand{\dopt}{\dist_{\mathrm{opt}}}
\DeclareMathOperator{\tw}{tw}
\DeclareMathOperator{\dom}{dom}

\DeclareMathOperator*{\EX}{\mathrm{E}}% expected value
\DeclareMathOperator{\Pos}{Pos}
\DeclareMathOperator{\Opt}{Opt}
\DeclareMathOperator{\Feas}{Feas}
\DeclareMathOperator{\val}{val}
\DeclareMathOperator{\minval}{minval}
\DeclareMathOperator{\maxval}{maxval}
\DeclareMathOperator{\tup}{tup}
\DeclareMathOperator{\ar}{ar}
\DeclareMathOperator{\phom}{p-hom}
\DeclareMathOperator{\supp}{supp}

\newcommand{\toset}[1]{\operatorname{Set} (#1)}
\newcommand{\SA}[1]{\operatorname{SA}_{#1}}

\newcommand{\BigO}[1]{\mathcal{O}\bc{#1}}
\newcommand{\NN}{\mathbb{N}}

\newcommand{\QQ}{\mathbb{Q}}
\newcommand{\ZZ}{\mathbb{Z}}
\newcommand{\Qinf}{\ensuremath{\QQ_{\geq 0} \cup \set{-\infty}}}
\newcommand{\Qposinf}{\ensuremath{\QQ_{\geq 0} \cup \set{\infty}}}
\newcommand{\Qpos}{\ensuremath{\QQ_{\geq 0}}}
\newcommand{\Eg}{\EX_{g\sim\omega}}
\newcommand{\eps}{\varepsilon}

\newcommand{\MaxSol}{\ensuremath{\mbox{Max-Sol}}\xspace}
\newcommand{\MinSol}{\ensuremath{\mbox{Min-Sol}}\xspace}
\newcommand{\StrictCSP}{\ensuremath{\mbox{Strict-CSP}}\xspace}
\newcommand{\MinCostHom}{\ensuremath{\mbox{Min-Cost-Hom}}\xspace}
\newcommand{\MaxClique}{\ensuremath{\mbox{Max-Clique}}\xspace}
\newcommand\ve[1]{\mathbf{#1}}
\newcommand\vx{\ve{x}}
\newcommand\vy{\ve{y}}

\newcommand{\Alg}{\mathtt{Algo}} % algorithm in 'efficiently Baker'.

\DeclarePairedDelimiter{\PDbc}{\lparen}{\rparen}
\newcommand{\bc}[1]{\PDbc*{#1}}
\DeclarePairedDelimiter{\PDbrk}{\lbrack}{\rbrack}
\newcommand{\brk}[1]{\PDbrk*{#1}}
\DeclarePairedDelimiter{\abs}{\lvert}{\rvert}
\DeclarePairedDelimiter{\card}{\lvert}{\rvert}
\DeclarePairedDelimiter{\set}{\lbrace}{\rbrace}
\newcommand{\I}[1]{\mathbf{1} \brk{#1}}
\newcommand\restr[2]{{% we make the whole thing an ordinary symbol
		\left.\kern-\nulldelimiterspace% automatically resize the bar with \right
		#1 % the function
		\right|_{#2} % this is the delimiter
}}

\newcommand{\tuple}[1]{\mathbf{#1}}
\DeclareMathOperator{\Gaifman}{G}
\DeclareMathOperator{\p}{p}

\newcommand\rel[1]{\mathrm{Rel}[#1]}
\newcommand{\istar}{{i^*}}

\begin{document}

\author{Balázs F. Mezei
\and
Marcin Wrochna\\
University of Warsaw\\
\texttt{m.wrochna@mimuw.edu.pl}
\and
Stanislav \v{Z}ivn\'y\\
University of Oxford\\
\texttt{standa.zivny@cs.ox.ac.uk}
}

\title{PTAS for Sparse General-Valued CSPs\thanks{
An extended abstract of this work appeared in~\emph{Proceedings of the 36th
Annual ACM/IEEE Symposium on Logic in Computer Science} (LICS 2021)~\cite{mwz21:lics}.
Stanislav \v{Z}ivn\'y was supported by a Royal Society University
Research Fellowship. Work mostly done while Balázs F. Mezei and Marcin
Wrochna were employed at the University of Oxford. This project has received funding from the European
Research Council (ERC) under the European Union's Horizon 2020 research and
innovation programme (grant agreement No 714532). The paper reflects only the
authors' views and not the views of the ERC or the European Commission. The
European Union is not liable for any use that may be made of the information
contained therein. This research was funded by UKRI EP/X024431/1. For the purpose of Open Access, the authors have applied a CC BY public copyright licence to any Author Accepted Manuscript version arising from this submission. All data is provided in full in the results section of this paper.}}

\maketitle

\begin{abstract}
We study polynomial-time approximation schemes (PTASes) for constraint satisfaction
problems (CSPs) such as Maximum Independent Set or Minimum Vertex Cover on sparse graph classes.

Baker's approach gives a PTAS on planar graphs, excluded-minor classes, and beyond.
For Max-CSPs, and even more generally, maximisation finite-valued CSPs (where
constraints are arbitrary non-negative functions), Romero, Wrochna, and
\v{Z}ivn\'y~[SODA'21] showed that the Sherali-Adams LP relaxation gives a simple PTAS for all fractionally-treewidth-fragile classes, which is the most general ``sparsity'' condition for which a PTAS is known.
We extend these results to general-valued CSPs, which include ``crisp'' (or ``strict'') constraints that have to be satisfied by every feasible assignment.
The only condition on the crisp constraints is that their domain contains an element which is at least as feasible as all the others (but possibly less valuable).

For minimisation general-valued CSPs with crisp constraints, we present a PTAS
for all \emph{Baker} graph classes --- a definition by Dvo\v{r}\'ak~[SODA'20] 
which encompasses all classes where Baker's technique is known to work, except for fractionally-treewidth-fragile classes.
While this is standard for problems satisfying a certain monotonicity condition on crisp constraints, we show this can be relaxed to \emph{diagonalisability} --- a property of relational structures connected to logics, statistical physics, and random CSPs.
\end{abstract}

\section{Introduction}\label{sec:intro}
Min-Ones and Max-Ones, studied by Khanna and Motwani (under the names of TMIN and
TMAX, respectively)~\cite{Khanna96:stoc} and by Khanna, Sudan, Trevisan, and
Williamson~\cite{Khanna01:sicomp}, are Boolean CSPs in which one seeks a
feasible solution (a 0--1 assignment satisfying all constraints) minimising/maximising the number of variables assigned the
label~$1$. Classical examples are the Minimum Vertex Cover and the Maximum
Independent Set problem, respectively. A natural generalisation to larger
alphabets is the problem in which one seeks a solution to a CSP instance while
minimising/maximising a sum of unary functions. With injective unary functions, such problems have
been studied under the name of \StrictCSP by  plumar, Manokaran, Tulsiani, and
Vishnoi~\cite{Kumar11:soda}, and Min/Max-Solution by Jonsson, Kuivinen, and
Nordh~\cite{Jonsson08:sicomp}. With arbitrary unary functions, such problems
have been studied under the name of \MinCostHom by Gutin, Hell, Rafiey, and
Yeo~\cite{Gutin08:ejc-mincostdichotomy}, Takhanov~\cite{Takhanov10:stacs}, and
others~\cite{Hell12:sidma,Hell12:esa-approximation,Rafiey19:icalp}.
In this paper we consider the still more general setting of \emph{general-valued} CSPs, where constraints are functions which give values to every possible assignment on a tuple of variables;
we allow $\infty$ or $-\infty$ values to express \emph{crisp} (also known as \emph{strict}) constraints,
which have to be satisfied by every feasible (finite-valued) assignment.
While a lot of research is devoted to exact algorithms or optimal approximation
ratios in APX-hard cases (see~\cite{JonssonN08,JeavonsKZ14,Makarychev17:survey} for surveys),
we seek the most general conditions that allow to obtain a polynomial-time approximation scheme (PTAS).

Baker~\cite{Baker94:jacm} gave an elegant method (sometimes known as the \emph{shifting} or \emph{layering technique}) for constructing polynomial-time approximation schemes (PTASes) which applies to many such problems, with the condition that the input instance's graph (the Gaifman graph) is ``sparse''.
This was initially presented for planar graphs, but it is known that similar structural properties are exhibited by all proper minor-closed graph classes~\cite{Grohe03,DeVos+04,Demaine05:focs} and beyond: e.g.\ graphs embeddable in a fixed surface with few intersections per edge~\cite{Marathe97,Grigoriev2007}, or sparse unit ball intersection graphs in few dimensions~\cite{Hunt98}
(but not e.g.\ 3-regular expanders: bounded degree is not sufficient to get a
PTAS even for Independent Set~\cite{Berman99:icalp}).
Dvo\v{r}\'ak~\cite{Dvorak16} defined \emph{fractionally-treewidth-fragile} classes
--- a natural generalisation of earlier sparsity conditions --- which encompasses all these examples.
A~class of graphs is fractionally-treewidth-fragile if one can remove vertices in a randomised way so that each vertex is removed with arbitrarily small probability $\eps$, but the treewidth after removal is always bounded, the bound depending on $\eps$ only.
He showed that if this notion of sparsity can be efficiently certified in a class of graphs, then this suffices to guarantee a PTAS, at least for a few problems such as Weighted Maximum Independent Set.
On the other hand it is not known whether this suffices for Minimum Vertex Cover, for example.

To remedy this, Dvo\v{r}\'ak~\cite{Dvorak20:soda} later defined \emph{Baker classes} and proved that (an effective version of) this condition suffices to provide a PTAS
to all monotone optimisation problems expressible in first-order logic (including of course Vertex Cover).
Very roughly, a class of graphs is \emph{Baker} if one can reduce each graph in it to the empty graph by a bounded number of the following steps: either remove a single vertex, or select a breadth-first-search layering and recurse into all subgraphs that can be induced by a few consecutive layers.
Dvo\v{r}\'ak proved that the family of Baker classes still includes all the examples discussed above;
on the other hand, it is strictly included in the family of fractionally-treewidth-fragile classes (and hence less general)~\cite{DvorakPersonal}.
It is worth mentioning that proper minor-closed graph classes can be shown to be
Baker (and fractionally-treewidth-fragile) relatively easily, without using the
Graph Minor Structure theorem, in contrast to the earlier, less general
definitions (see~\cite{Dvorak20:soda} for details).

In order to provide a PTAS for a class of general-valued CSPs, \emph{a sparsity
condition is not enough}:
we also need to restrict what types of constraints can be used in an instance. 
Otherwise, even if the values to be optimised are trivial, either 0 or infinity,
one could use the crisp constraints to express 3-Colouring, which is NP-hard
even on planar graphs of bounded degree~\cite{Dailey80}. In fact, as long as all
crisp constraints are available, for any possible restriction on Gaifman graphs,
either the restriction implies bounded treewidth, making the problem exactly
solvable, or it is hard to decide whether the optimum is zero or infinite, by a
result of Grohe, Schwentick, and Segoufin~\cite{GroheSS01}.
We will hence require a condition which ensures that one can easily decide whether a feasible solution (of finite value) exists.
This usually takes the form of a monotonicity condition.

On the other hand, \emph{some sparsity condition is also necessary}: on
general Gaifman graphs, there is no restriction of constraint types that
would result in a general-valued CSP that admits a PTAS but is not solvable exactly in
polynomial time.\footnote{This follows from the NP-hardness
result of Kozik and Ochremiak~\cite{Kozik15:icalp}, which actually shows
APX-hardness; for earlier, explicit APX-hardness results for CSPs see, e.g.,~\cite{Jonsson08:sicomp,Jonsson09:tcs}.
However, we remark non-trivial PTAS examples are known for ``surjective'' maximisation finite-valued CSPs~\cite{FullaUZ19}.}
In this sense our work follows the line of ``uniform'' or ``hybrid'' CSPs, which
include restrictions on both the input's Gaifman graph (\emph{left-hand side}
restrictions)  and on the types of constraints (\emph{right-hand side} restrictions); see~\cite{CooperZ17} for a survey.
However, unlike that line of work, we look for PTASes instead of exact solvability, which also lets us go well beyond planar graphs and beyond very specialised algebraic algorithms.

\subsection{Related work}
The exact solvability of general-valued CSPs has
been characterised for left-hand side restrictions~\cite{crz22:sicomp} (tractable cases are precisely classes that have bounded treewidth, up to a certain notion of homomorphic equivalence) and right-hand side
restrictions~\cite{Kolmogorov17:sicomp} (tractable cases are precisely
delineated by certain algebraic properties); both results include the case where infinite values are allowed.

As discussed above, there are no PTASes for general-valued CSPs with only left-hand side or only right-hand side restrictions, beyond exactly solvable cases.
In fact Khanna et al.~\cite{Khanna01:sicomp}, in their work on Min-Ones and Max-Ones with right-hand side restriction, remark that ``Our framework lacks such phenomena as PTAS'' and discuss left-hand side restrictions as an interesting avenue for future work for that reason.
Similarly~\cite{Jonsson08:sicomp} and~\cite{JonssonN08} ask in the context of
right-hand side restricted Min-Solution and Max-Solution problems: ``Under which restrictions on variable scopes does Max Sol admit a PTAS?''. 

Very recently, PTASes for left-hand side restricted Max-CSP without crisp constraints, such as Max-Cut, have been studied by Romero, Wrochna, and \v{Z}ivn\'y~\cite{rwz21:soda}. 
More generally, they consider so-called finite-valued CSPs, where the only right-hand side restriction is having finite, non-negative values.
They showed a PTAS is possible for every fractionally-treewidth-fragile class of Gaifman graphs.
In fact the algorithm is simply the Sherali-Adams linear programming relaxation (with a growing number of levels giving a better and better approximation),
which is oblivious to the graph structure and does not require it to be efficiently certified in any way.

As for constant-factor approximations, Raghavendra's celebrated result gave the
best approximation ratio, assuming the Unique Games Conjecture of Khot~\cite{Khot02stoc}, for all
right-hand side restricted Max-CSPs (and also finite-valued
CSPs)~\cite{Raghavendra08:everycsp}.
Analogous results for monotone Strict-CSPs were obtained by Kumar et
al.~\cite{Kumar11:soda}. Constant-factor approximation algorithms have been
established for right-hand side restricted \MinCostHom on special graphs by
Hell, Mastrolilli, Nevisi, and Rafiey~\cite{Hell12:esa-approximation}, and for
all graphs and some digraphs by Rafiey, Rafiey, and Santos~\cite{Rafiey19:icalp}.

\subsection{Our results}
As in most research on constraint satisfaction, our main motivation is
to understand the mathematical structure that governs efficient computation and,
if at all possible, to obtain complexity classifications of large fragments of
CSPs. In this paper, the goal was to push the tractability frontier of
general-valued CSPs that admit polynomial-time approximation schemes.
In particular, we try to understand what right-hand side restrictions make
Baker's technique possible, tentatively answering the aforementioned questions
from~\cite{Khanna01:sicomp,Jonsson08:sicomp,JonssonN08}. We show that the most
general results known for Vertex Cover and Independent Set (PTASes on all Baker
or fractionally-treewidth-fragile classes, respectively) can be extended to any
to general-valued CSPs with a certain monotonicity restriction, and even further.

To clearly separate left-hand side and right-hand side restriction,
it is convenient to phrase a general-valued CSP (VCSP) as the problem of optimising the
value of a function between two valued structures.
Precise definitions are given in \cref{sec:prelims}.
Briefly, a \emph{valued structure} $\A$ consists of a (finite) domain $A$ and a collection of functions $f^\A \colon A^n \to \QQ \cup \{\pm \infty\}$, indexed by symbols $f$ belonging to a set of symbols $\sigma$ called a signature.
For two (finite) structures $\A,\C$, the \emph{value} of an assignment $h \colon A \to C$
is an expression of the form $\sum f^\A(\vx) f^\C(h(\vx))$.
We will be seeking to find either the minimum or maximum value over all assignments,
denoted $\minval(\A,\C)$ and $\maxval(\A,\C)$ respectively.
\emph{Feasible} assignments are those of finite value.
The reader should think of the left-hand side structure $\A$ as of a set of variables $A$ together with weighted constraint scopes: for $\vx \in A^n$, $f^\A(\vx) = w \neq 0$ means that the instance applies the constraint ``$f$'' to variables in $\vx$ with weight $w$. 
The right-hand side structure $\C$ encodes the alphabet $C$ (to which an assignment $h$ maps each variable) and the collection of available constraints, which could be arbitrary $\QQ \cup \{\pm \infty\}$-valued cost functions in general.
An instance of the VCSP is a pair $(\A,\C)$; its Gaifman graph, denoted
by $\Gaifman(\A)$, is a graph whose vertex set is the domain $A$ with edges between two vertices that occur together in a constraint of non-zero weight.

\subsection*{Minimum Solution}
For minimisation, we first consider $(\QQ_{\geq 0} \cup \{\infty\})$-valued right-hand side structures $\C$, in which the sets of zero-valued tuples and finite-valued tuples are anti-monotone, in the following sense.
There is a total order $\leq_\top$ on $C$, and for all tuples $\vx, \vy \in C^n$ with $\vx \leq_\top \vy$ (coordinate-wise) we have that for all non-unary function symbols $f$ of $\C$:
\begin{itemize}
	\item $f^\C(\vx) < \infty$ implies $f^\C(\vy) < \infty$, and
	\item $f^\C(\vx) = 0$ implies $f^\C(\vy) = 0$.
\end{itemize}
Intuitively, larger tuples are more feasible.
We call valued structures $\C$ satisfying this condition $\MinSol$ structures.
We define $\MinSol_\cG$ to be the general-valued CSP restricted to instances
$(\A,\C)$ where $\A$ is a $\QQ_{\geq 0}$-valued structured with $\Gaifman(\A) \in \cG$ and $\C$ is a \MinSol structure.

For example, Weighted Minimum Vertex Cover is equivalent to the $\MinSol$ case where 
$\C$ is the structure with domain $\{0,1\}$ and $0\leq 1$ and
with a 2-ary cost function $f^\C(0,0)=\infty$, $f^\C(1,0)=f^\C(0,1)=f^\C(1,1)=0$,
and a unary cost function $u^\C(0) = 0$, $u^\C(1) = 1$.

We
show that $\MinSol_\cG$ admits a PTAS for all graph classes $\cG$ that are efficiently Baker.
(Dvo\v{r}\'{a}k's definition is somewhat involved but we give an exposition in \cref{sec:baker}).
As discussed above, this captures essentially all graph classes where a version of Baker's technique is known to apply (including excluded-minor classes and more),
except for fractionally-treewidth-fragile classes.
We remark that already the very special case of Minimum Vertex Cover is not known to admit a PTAS on fractionally-treewidth-fragile classes.

Simultaneously, our results are less restrictive on the right-hand side,
as unlike in earlier work such as the framework of \StrictCSP of~\cite{Kumar11:soda},
we allow arbitrary values strictly between $0$ and $\infty$ (not only on unary constraints).
Once we realise this is possible, however, the algorithm turns out to be a rather standard application of Baker's technique: the only difference is that we increase the number of layers to account for the maximum ratio between finite, positive values (which is a constant depending on values of $\C$ only).

The main novelty in our work is establishing the existence of a
PTAS under a weaker assumption on the right-hand side structure $\C$ -- we only
require that $\C$ should be a diagonalisable structure. (As we will
show in \cref{lem:minsol-dis}, all
$\MinSol$ structures are diagonalisable and thus our result establishes a PTAS
for $\MinSol$ structures as a special case.)
Diagonalisability is a notion derived from the work of Brightwell and
Winkler~\cite{BrightwellW00} in the case of graphs and Brice{\~n}o, Bulatov,
Dalmau, and Larose~\cite{Briceno21:jctb} in the case of 
relational structures (which are more general than graphs).
The precise definition of diagonalisability is technical and can be found
in \cref{sec:dism}.
For relational structures, one characterisation is that a structure $\C$ is diagonalisable if and only if
the two projection homomorphisms $\pi_1, \pi_2 \colon \C \times \C \to \C$ (defined as $\pi_i(x_1,x_2) = x_i$)
are connected by some sequence of homomorphisms $\psi \colon \C \times \C \to \C$ such that consecutive homomorphisms in the sequence differ at only one vertex, and all the homomorphisms in the sequence are idempotent (meaning $\psi(x,x)=x$).
This turns out to be equivalent to saying that for all structures $\A$, the set of all homomorphisms from $\A$ to $\C$ is connected in a similar sense.
A few other characterisations connect diagonalisability to statistical physics via ``mixing'' properties.
Diagonalisability is also equivalent to finite duality (the existence of finitely many obstructions to having a homomorphism into $\C$),
a notion important to the study of CSPs via logic~\cite{BulatovKL08}.
For these and many other equivalent definitions of diagonalisability, cf.~\cite[Corollary~6.3 and Theorem~3.6]{Briceno21:jctb} with $J=V(H)$.

Our main result for minimisation 
(proved in \cref{sec:baker})
is an approximation scheme for instances $(\A,\C)$ where $\A$ comes from a Baker class and $\C$ is diagonalisable.
One should think of the functions $f_1$ and $f_2$ as polynomials depending on the size of the graph $\Gaifman(\A)$.

\newcommand{\StateMinSolBaker}{
Let $\cG$ be an $(f_1,f_2)$-efficiently Baker class.
Then, for any $\eps>0$ and any instance $(\A,\C)$ of general-valued CSP
where $\A$ is a $\QQ_{\geq 0}$-valued structured with $\Gaifman(\A) \in \cG$ and $\C$ is diagonalisable, we can find a solution of value
at most $(1+\eps) \minval(\A,\C)$ in time $f_1(\lvert\A\rvert) + f_2(c\lvert\A\rvert) \cdot c^{1/\eps}$
where $c$ depends on $\C$ and $\cG$ only.
}
\begin{theorem}\label{thm:minhomBaker}
\StateMinSolBaker%
\end{theorem}

Here the constant $c$ depends polynomially on $\lvert \C \rvert$ and exponentially on the maximum ratio between certain finite positive values of $\C$.
Since every class of graphs that excludes a minor is
$(\BigO{n^2},\BigO{n})$-efficiently Baker~\cite[Theorem~2.1]{Dvorak20:soda}, \cref{thm:minhomBaker} in fact gives an EPTAS on such classes for any fixed diagonalisable structure $\C$.

Intuitively, diagonalisability allows to interpolate between any two homomorphisms, and we show this gives a natural way to combine partial solutions in the way needed in Baker's technique (generalising the simple combination used for Vertex Cover: taking the set-theoretic sum of solutions).
This proof (Theorem~\ref{thm:minhom}), which is an entirely new connection between diagonalisability and approximation, is our main contribution.

\subsection*{Maximum Solution}
For maximisation, we extend the results of~\cite{rwz21:soda},
which restricted the right-hand side $\C$ to be $\QQ_{\geq 0}$-valued.
We additionally allow $-\infty$ values,
but the set of tuples $\vy \in C^n$ with $f^\C(\vy) = -\infty$
is restricted to be monotone in the following very weak sense.
There is an element $c_\bot \in C$ such that
whenever $\vy$ is feasible ($f^\C(\vy) \neq -\infty$)
and $\vy'$ is a tuple obtained from $\vy$ by replacing some of its elements with $c_\bot$, 
then $\vy'$ is still feasible ($f^\C(\vy') \neq -\infty$).

We call structures $\C$ satisfying this condition $\MaxSol$ structures
and we define $\MaxSol_\cG$ to be the general-valued CSP restricted to instances
$(\A,\C)$ where  $\A$ is a $\QQ_{\geq 0}$-valued structured with  $\Gaifman(\A) \in \cG$ and $\C$ is a \MaxSol structure.

For example, Weighted Maximum Independent Set is equivalent to the $\MaxSol$ case where 
$\C$ is the structure with domain $\{0,1\}$,
with a 2-ary function $f^\C(1,1)=-\infty$, $f^\C(0,0)=f^\C(1,0)=f^\C(0,1)=0$,
and a unary function $u^\C(0) = 0$, $u^\C(1) = 1$ (so $c_\bot=0$).

Our main result for maximisation 
(proved in \cref{sec:maxsol}) is the following.

\begin{theorem}\label{thm:main-maxsol}
  Let $\cG$ be a class of graphs that is fractionally-treewidth-fragile.
  Then $\MaxSol_\cG$ admits a PTAS\@.
  
  More precisely, for all $\eps>0$, there is an algorithm that given $(\A,\C)$,
  outputs a value between $\maxval(\A,\C)$ and $(1+\eps) \cdot \maxval(\A,\C)$
  in time ${(|\A| + |\C|)}^{k(\eps)}$, where $k(\eps)$ is a function depending on $\cG$ only.\footnote{
    If $\cG$ is fractionally-$\tw$-fragile with rate $k(\eps)$ (Definition~\ref{def:fragile}),
    then the exponent in the running time is $\BigO{k(\Omega(\eps))}$.
    See~\cite{DvorakS20} for a fragility rates of various graph classes.
  }
\end{theorem}

The algorithm in \cref{thm:main-maxsol} does nothing more than solve a
$\Theta(k(\eps))$-th level of the Sherali-Adams linear programming relaxation.
This allows the algorithm to be oblivious to the graph structure, i.e.\ we do not assume that the fractional-treewidth-fragility of $\cG$ can be efficiently certified.
Thus the left-hand side restriction on Gaifman graphs is the most general for which a PTAS is known;
as discussed earlier, it includes excluded-minor classes and more. 
In fact similarly to~\cite{rwz21:soda}, we conjecture that $\MaxSol_\cG$ does not admit a PTAS for any $\cG$ that is not fractionally-treewidth-fragile.
Since $\MaxSol_\cG$ is strictly more general (by allowing negative infinite values), this conjecture might be easier to prove than the one in~\cite{rwz21:soda}.

On the other hand, this approach does not give an EPTAS even when $\C$ is fixed (i.e.\ the exponent of $|\A|$ increases with $\eps$), and it does not construct an assignment --- it only approximates the optimum value.
In contrast, given a class of graphs $\cG$ for which fractional-treewidth-fragility can be efficiently certified (which includes essentially all known examples),
it is straightforward to construct solutions to $\MaxSol_\cG$ of value at least $(1-\eps) \cdot \maxval(\A,\C)$ in time $\lvert\A\rvert \cdot {\lvert\C\rvert}^{k(\eps)}$.

Our main contribution in proving \cref{thm:main-maxsol} is finding the right analogues of the definitions from~\cite{rwz21:soda} -- a notion of ``closeness'' of structures, a dual notion that certifies this closeness with concrete mappings (a distribution of ``partial homomorphisms'', see Section~\ref{sec:maxsol}), and the proof of their equivalence (Lemma~\ref{lem:duality}).
In particular, while the name ``partial homomorphism'' may sound deceptively simple, we found that pin-pointing their definition (in the context of MaxSol) proved to be a surprisingly intricate balancing act.

We complement \cref{thm:main-maxsol} with simple constructions which show that it is impossible to extend other results of~\cite{rwz21:soda} from the setting of purely optimisation Max-CSPs to the setting of general-valued CSPs, which include crisp constraints.
In~\cite{rwz21:soda} the notion of \emph{pliability} is defined (for Max-CSPs), which is a left-hand side restriction that takes the whole structure~$\A$ into account, not only its Gaifman graph, as done in this introduction so far;
this allowed the authors of~\cite{rwz21:soda} to show that the same framework applies not only to sparse, fractionally-treewidth-fragile instances of Max-CSPS, but also to dense structures.
We define an analogous notion of \emph{strong pliability} and show in \cref{lem:tract}, similarly
to~\cite{rwz21:soda}, the existence of a PTAS (for general-valued CSPs) under the strong pliability
assumption on the left-hand side structure, which takes the whole structure
$\A$ into account, not only its Gaifman graph. (Thus this is a more general
tractability result than \cref{thm:main-maxsol}.) 
However, in \cref{sec:dense} 
we show that even the simplest class of dense structures, namely the class of
$\{0,1\}$-valued cliques, does not satisfy strong pliability.
In fact, it is easy to show (cf.~\cref{sec:dense}) that the $\MaxSol$ problem is hard to approximate even when the left-hand side structures are restricted to cliques.

\paragraph*{Paper organisation}
\cref{sec:prelims} introduces basic notations and defines the studied computational problems.
The main result for minimisation, \cref{thm:minhomBaker}, is technical and proved 
in \cref{sec:baker}.
In \cref{sec:minsol}, we present the main ideas in the special case of
planar structures. The main result for maximisation, \cref{thm:main-maxsol}, is proved
in \cref{sec:maxsol}. Some of the proofs are deferred 
to \cref{sec:overcast-sa,sec:SA-tw,sec:dense}.

\section{Preliminaries}\label{sec:prelims}
For an integer $k$, we denote by $[k]$ the set $\{1,\ldots,k\}$. For a tuple
$\vx$, we denote by $x_i$ its $i$-th coordinate and by $\toset{\vx}$ the set of
elements appearing in $\vx$. For two tuples $\vx$ and $\vy$ of length $n$, we
write $(\vx,\vy)$ as a shorthand for $((x_1,y_1),(x_2,y_2),\dots,(x_n,y_n))$.
For a tuple $\vx$ of length $n$ and a map $h$, we denote by $h(\vx)$ the
coordinate-wise application of $h$; i.e., $h(\vx)=(h(x_1),\ldots,h(x_n))$.

\paragraph*{General-valued CSPs}
A \emph{signature} is a finite set $\sigma$ of (function) symbols such as $f$, each with
a specified arity $\ar(f)$. 
For a set of values $\Omega \subseteq\QQ \cup \{-\infty,+\infty\}$,
an \emph{$\Omega$-valued structure} $\A$ over a signature $\sigma$ (or
$\sigma$-structure, for short) is a finite domain $A$ together with
a function $f^{\A}: A^{\ar(f)} \to \Omega$ for each symbol $f \in
\sigma$. 
We denote by $A,B,C,\dots$ the domains of structures $\A,\B,\C,\dots$.

We define $\tup(\A)$ to be the set of all pairs $(f,\tuple{x})$ such that $f \in
\sigma$ and $\tuple{x} \in A^{\ar(f)}$;
 and $\tup_{>0}(\A)$ to be the set of all pairs $(f,\tuple{x})\in\tup(\A)$ with $f^{\A}(\tuple{x})>0$.
 
We assume a straightforward \emph{table encoding} of structures: the interpretation $f^{\A}$ of a symbol $f$ in a structure $\A$ is encoded as a collection of triples $\{ (f, \tuple{x}, f^{\A}(\tuple{x})) \mid (f,\tuple{x}) \in \tup(\A)\}$. Thus, the size of a $\sigma$-structure $\A$ is roughly
\[
  |\A|=|\sigma|+|A|+\sum_{(f,\tuple{x})\in\tup(\A)} \log{|\sigma|}+\ar(f)\log{|A|}+|enc(f^{\A}(\tuple{x}))|
\]
where $enc(\cdot)$ denotes a reasonable encoding for elements of $\QQ$. 
 
We consider the following computational problem.
 
\begin{definition}
   An instance of the \emph{general-valued CSP} (VCSP) consists of an ordered pair of $\sigma$-structures $(\A,\C)$.
   For a mapping $h\colon A \to C$, we define the value of $h$ to be
  \begin{align*}
 	\val(h) = \sum_{(f,\vx) \in \tup(\A)} \fA(\vx) \fC(h(\vx)).
  \end{align*}
  The goal is to find the minimum or maximum value over all possible mappings
  $h\colon A\to C$,
  denoted $\minval(\A,\C)$ or $\maxval(\A,\C)$, respectively.
\end{definition}
 
On the left-hand side we will only use $\Qpos$-valued structures, with letters $\A,\B$;
on the right-hand side we will only use $\Qposinf$ or $\Qinf$-valued structures, respectively, for minimisation and maximisation, with letters $\C,\D$.

For $\lambda\geq 0$ we write $\lambda\A$ for the \emph{rescaled} $\sigma$-structure with domain $A$ and $f^{\lambda\A}(\tuple{x}) := \lambda f^{\A}(\tuple{x})$, for $(f,\tuple{x}) \in \tup(\A)$.
For a $\sigma$-structure $\A$ and subset of the domain $X \subseteq A$, we
define $\A[X]$ to be the restriction of $\A$ to $X$. That is, $\A[X]$ is a
$\sigma$-structure over the domain $X$, and $f^{\A[X]}(\vx) = \fA(\vx)$ for each
$f \in \sigma$ and $\vx \in X^{\ar(f)}$.

Following the influential work on decision CSPs by Grohe, Schwentick, and
Segoufin~\cite{GroheSS01}, and Grohe~\cite{Grohe07:jacm}, we will
focus on fragments of the VCSP parametrised by the class of
left-hand side structures (or their underlying class of
graphs). Given a $\sigma$-structure $\A$, the \emph{Gaifman graph} (or
\emph{primal graph}), denoted by $\Gaifman(\A)$, is the graph whose vertex set
is the domain $A$, and whose edges are the pairs $\{u,v\}$ for which there is a
tuple $\tuple{x}$ and a symbol $f\in \sigma$ such that $u,v$ appear in
$\tuple{x}$ and $f^{\A}(\tuple{x})>0$. 

For a graph parameter $\p$ and a structure $\A$, we define
$\p(\A)\coloneqq\p(\Gaifman(\A))$ to be the parameter of the Gaifman graph of
$\A$. In particular, the \emph{treewidth} of $\A$ is defined as
$\tw(\A)\coloneqq\tw(\Gaifman(\A))$. (We will only use treewidth and excluded
minors as black-boxes and thus will not need their definitions. The reader is referred to Diestel's
textbook for details~\cite{Diestel10:graph}.)

\paragraph*{Relational structures}
A \emph{relational $\sigma$-structure} $\C$ includes for each symbol
$f\in\sigma$ a relation $\fC \subseteq C^{\ar(f)}$. We will view relational
structures as $\set{0,\infty}$-valued structures by associating each function
$\fC:C^{\ar(f)} \to \set{0,\infty}$ to the relation given by the zero-valued
tuples $\set{\vx \mid \fC(\vx)=0}$.
A \emph{homomorphism} from a relational $\sigma$-structure~$\C$ to a relational
$\sigma$-structure~$\D$ is a map $\psi:C\to D$ that satisfies, for every
$f\in\sigma$ and every $\vx\in C^{\ar(f)}$, $f^{\D}(\psi(\vx))\leq f^{\C}(\vx)$.

For an $n$-ary function $f$, we denote by $\Feas(f)$ and
$\Opt(f)$ the $n$-ary relations defined by $\Feas(f)=\{\vx\mid f(\vx)<\infty\}$
and $\Opt(f)=\{\vx\mid f(\vx)=0\}$, respectively.
Let $\C$ be a $\sigma$-structure. The relational $\sigma$-structure $\Feas(\C)$ contains, for each $f\in\sigma$, the relation $\Feas(\fC)$; similarly, the relational
$\sigma$-structure $\Opt(\C)$ contains, for every $f\in\sigma$, the relation
$\Opt(\fC)$.

Our results will be concerned with two particular types of right-hand side structures.

\subsection*{Maximum Solution}

For the following definition, recall the example of Weighted Maximum
Independent Set from Section~\ref{sec:intro}. One should think of the element
$c$ in this context as not selecting a vertex in an independent set; the partial
order on $C^n$ then says that a subset of an independent set is also independent.

\begin{definition}[{$\poleq[c]$}]
For an element $c$ of a set $C$, we denote by $\poleq[c]$ the partial ordering on $C$
defined by $c \poleq[c] x$ and $x\poleq[c] x$ for all $x\in C$.
This induces a partial ordering on $C^n$ coordinate-wise:
we write $\vx\poleq[c]\vy$ for $\vx,\vy \in C^n$ if we can obtain $\vx$ from~$\vy$ by changing some (possibly none or all) of its coordinates to $c$.
\end{definition}

\begin{definition}[\MaxSol]\label{def:maxsol}
  Let $\sigma$ be a finite signature.
  A  $\sigma$-structure $\C$ is called a \emph{\MaxSol structure}
  if it is $(\Qinf)$-valued and
  there is an element $c_\bot \in C$ such
  that for all $f\in\sigma$, the following holds:
      whenever $f^{\C}(\vy) \geq 0$, we have $f^{\C}(\vx) \geq 0$, for all $\vx \poleq[c_\bot] \vy$ in $C^{\ar(f)}$.
  Equivalently, if a tuple $\vy$ has non-negative value (not $-\infty$),
  then changing some of its coordinates to $c_\bot$ still gives a non-negative
  value.
  To avoid clutter, we write $\poleq$ in place of $\poleq[c_\bot]$,
  with the choice of $c_\bot \in C$ implicit.

  We denote by $\MaxSol_\cG$ the restriction of the VCSP to instances $(\A,\C)$
  where $\A$ is a $\QQ_{\geq 0}$-valued structure with $\Gaifman(\A) \in \cG$
  and $\C$ is a \MaxSol structure.
\end{definition}

Observe that every $\Qpos$-valued structure is a \MaxSol structure;
thus $\MaxSol_\cG$ is more general than the restriction to $\Qpos$-valued
right-hand side structures, which is the problem considered in~\cite{rwz21:soda}.

\begin{example}
As explained in the introduction, \MaxSol structures can capture problems such as Weighted Maximum Independent Set.
Another example is finding a 3-colourable induced subgraph with the maximum number of edges:
take $\C$ with domain $C = \{R,G,B,\bot\}$ (representing red, green, blue, as well as a fourth element corresponding to vertices not selected into the induced subgraph) and a single symbol of arity two with values $f^\C(R,R)=f^\C(G,G)=f^\C(B,B)=-\infty$, $f^\C(x,y)=1$ for $x\neq y \in \{R,G,B\}$ and $f^\C(x,y)=0$ otherwise.
This extends to maximising the number of vertices (by introducing a unary relation), to weighted versions (by giving weights to vertices and edges of the left-hand-side structure $\A$),  and to finding a maximum induced substructure satisfying an arbitrary CSP.
\end{example}

\begin{remark}
  The ``downward monotone \StrictCSP'' from~\cite{Kumar11:soda} corresponds
  to \cref{def:maxsol} with some extra conditions. Firstly, there is a special
  unary symbol $u \in \sigma$ such that $u^{\C}$ is $\Qpos$ valued and all other symbols $f \in \sigma$ are $\{0,-\infty\}$-valued (hence they express ``strict'' constraints). Secondly, there is a total order on $C$, and for each symbol $f \in \sigma$ other than $u$, $f^\C$ is anti-monotone; in other words, lowering some coordinates of a tuple in $C^{\ar(f)}$ can not change its value from~$0$~to~$-\infty$. (Hence the minimum element plays the role of the bottom label $c_{\bot}\in C$.)
\end{remark}

\subsection*{Minimum Solution}

\begin{definition}[\MinSol]\label{def:minsol}
Let $\sigma$ be a finite signature.
A $\sigma$-structure $\C$ is called \MinSol if
it is $(\Qposinf)$-valued and
there is a total order $\leq_\top$ on $C$ such that:
for all $f\in \sigma$ with $\ar(f) > 1$ and 
all tuples $\vx, \vy \in C^n$ with $\vx \leq_\top \vy$ (coordinate-wise) we have:
\begin{itemize}
	\item $f^\C(\vx) < \infty$ implies $f^\C(\vy) < \infty$, and
	\item $f^\C(\vx) = 0$ implies $f^\C(\vy) = 0$.
\end{itemize}
  We denote by 
  $\MinSol_\cG$ the restriction of the VCSP to instances $(\A,\C)$
  where $\A$ is a $\QQ_{\geq 0}$-valued structure with $\Gaifman(\A) \in \cG$ and $\C$ is a \MinSol structure.
\end{definition}

\begin{remark}
The ``upward monotone \StrictCSP'' from~\cite{Kumar11:soda} corresponds to \cref{def:minsol} with the extra conditions that there is only one unary symbol $u$, $u^{\C}$ is monotone and injective, and all other cost functions $f^{\C}$ are $\{0,\infty\}$-valued (hence they express ``strict'' constraints).
\end{remark}

\begin{remark}
We observe that \emph{some} structure (such as a total order) on
the domain of a right-hand side \MinSol structure is needed: We show how
to encode 3-Colouring of planar graphs, which does not admit a
  PTAS (assuming P$\neq$NP).

Let $G$ be a planar graph. Let $\A$ be a structure with domain $V(G)$ over
the signature $\sigma = \set{u,f}$ of arities 1 and 2, respectively. Let
  $\uA(x)=1$ for all $x \in V(G)$, and $\fA(x,y) = 1$ if $\{x,y\} \in E(G)$ and $0$
otherwise. Let $\C$ be a right-hand side structure with domain
$C=\set{R,G,B,c_\top}$. Here we think of $R,G,B$ as three colours, and
$c_\top$ as a fourth extra colour we want to avoid using. We allow a
monochromatic $c_\top$ edge. Let $\uC(x) = 1$ for $x=c_\top$ and $0$
otherwise; $\fC(R,R)=\fC(G,G)=\fC(B,B)=\infty$, and $0$ for other pairs of
values (including $(c_\top, c_\top)$). If $G$ is 3-colourable then 
$\minval(\A, \C) = 0$. Otherwise, $\minval(\A, \C) \geq 1$.
Note that $\fC$ respects the partial order $\poleq[c_\top]$, but it does
not respect any total order on $C$. 
\end{remark}

\section{Minimisation on planar structures}\label{sec:minsol}
\subsection{Diagonalisability}\label{sec:dism}

Brice{\~n}o, Bulatov, Dalmau, and Larose defined the concepts of product
structure, dismantlability, adjacency, and
link graph for relational structures~\cite{Briceno21:jctb}. In this section, we
will extend these concepts to valued structures in a natural way. In particular, our
definitions (for structures) coincide with the definitions
in~\cite{Briceno21:jctb} (for relational structures) when viewed as $\set{0,\infty}$-valued structures.

Informally, we shall consider any two positive finite values to be basically equivalent, because we will be able to bound differences between them by a constant factor; so we shall consider an increase in value significant only if it increases from zero to positive or from finite to infinite.
For a structure $\C$, we say an element $a \in C$ is \emph{dominated} by an element $b \in C$ if we can always replace $a$ with $b$: for any mapping $h \colon A \to C$ (from any structure $\A$), assignments to $a$ can be changed to assignments to $b$ without increasing $\val(h)$ significantly (from zero to positive or from finite to infinite).
A structure $\C$ is \emph{diagonalisable} if in the product $\C \times \C$ (defined below), one can remove dominated elements one by one until only the diagonal $\{(c,c) \mid c \in C\}$ is left.
We will later see how this allows to ``blend in'' two different mappings $h_1, h_2$ from $\A$ to $\C$ by considering them together as a mapping to $\C^2$.

We now proceed with formal definitions.
Given two (valued) $\sigma$-structures $\C$ and $\D$, we call
$\psi:C \to D$ a \emph{homomorphism} from $\C$ to $\D$ if $\psi$ is a
homomorphism from $\Feas(\C)$ to $\Feas(\D)$ and from $\Opt(\C)$ to $\Opt(\D)$ (in other words, finite-valued tuples map to finite-valued tuples and zero-valued tuples map to zero-valued tuples).
It will be more convenient to consider both the $\Feas(\C)$ and $\Opt(\C)$
simultaneously. Thus with every structure $\C$ we will associate a relational
structure $\rel\C$, defined as follows. 

\begin{definition}\label{def:relval}
Let $\sigma$ be a valued signature. For any $f\in\sigma$, we denote by $f_1$ and $f_2$ two new relational symbols of the same arity as $f$.
Let $\C$ be a $\sigma$-structure and let $\sigma'=\bigcup_{f \in
\sigma}\set{f_1,f_2}$.
Define the
relational $\sigma'$-structure $\rel\C$ over the domain $C$ as follows: for each $f \in \sigma$, let $f_1^{\rel\C} = \Feas(f^\C) = \{\vx\mid f^\C(\vx)<\infty\}$ and $f_2^{\rel\C}  = \Opt(f^\C)=\{\vx\mid f^\C(\vx)=0\}$.
\end{definition}

We can now define the concepts of interest for structures $\C$ via the already
existing concepts for relational structures $\rel\C$
from~\cite{Briceno21:jctb}.
We use the following observation.

\begin{observation}
	For $x,y \in \QQ_{\geq 0} \cup \{\infty\}$, there exists $M > 0$ such that $y \leq M \cdot x$ if and only if:
	\begin{itemize}
	\item if $x < \infty$, then $y < \infty$, and
	\item if $x = 0$, then $y = 0$.
	\end{itemize}
\end{observation}

Given $\sigma$-structures $\C$ and $\D$, we say that $\psi:C \to D$ is a \emph{homomorphism}
if $\psi$ is a homomorphism from $\rel\C$ to $\rel\D$. Equivalently, $\psi$ is
a homomorphism if there exists $M>0$ such that for all $(f,\vx)\in\tup(\C)$,
\[ \fD(\psi(\vx)) \leq M \cdot \fC(\vx). \]
Here we can use a uniform bound $M$ because we only work with finite structures;
it will be convenient to use this equivalent definition to keep track of the bound $M$.

Given $\sigma$-structures $\C$ and $\D$ we define the \emph{product structure} $\C\times\D$ as a $\sigma$-structure with domain $C \times D$ and for each $f\in\sigma$,
\[ f^{\C\times\D}((\vx,\vy)) =  \fC(\vx) + \fD(\vy). \]
Let $\pi_1(x,y)=x$ and $\pi_2(x,y)=y$ be the projections to the first and second
coordinate, respectively. Note that $\pi_1,\pi_2$ are homomorphisms from $\C^2$ to $\C$ for any $\C$.
See Figure~\ref{fig:product}.

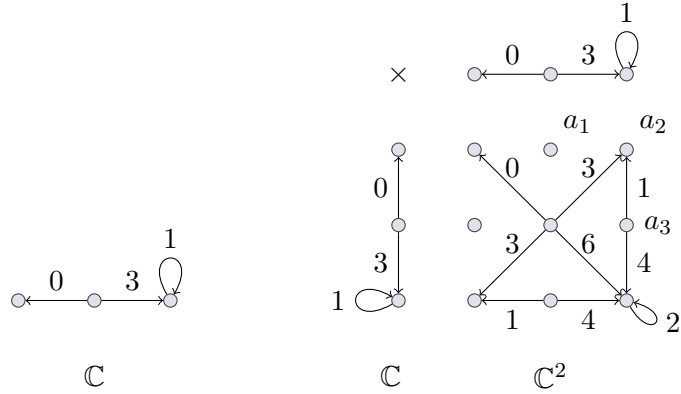
\begin{figure}[h!]
\centering
\begin{tikzpicture}
\node[v] (x) at (-1,0) {};	
\node[v] (y) at (0,0) {};	
\node[v] (z) at (1,0) {};	
\draw[->] (y) --node[midway,above]{$0$} (x);
\draw[->] (y) --node[midway,above]{$3$} (z);
\draw[->] (z) to[looseness=20,in=60,out=120] node[midway,above]{$1$} (z);
\node at (0, -1) {$\mathbb{C}$};

\begin{scope}[shift={(4,3)}]
\node[v] (x1) at (1,0) {};	
\node[v] (y1) at (2,0) {};	
\node[v] (z1) at (3,0) {};	
\draw[->] (y1) --node[midway,above]{$0$} (x1);
\draw[->] (y1) --node[midway,above]{$3$} (z1);
\draw[->] (z1) to[looseness=20,in=60,out=120] node[midway,above]{$1$} (z1);

\node[v] (x2) at (0, -1) {};	
\node[v] (y2) at (0, -2) {};	
\node[v] (z2) at (0, -3) {};	
\draw[->] (y2) --node[midway,left]{$0$} (x2);
\draw[->] (y2) --node[midway,left]{$3$} (z2);
\draw[->] (z2) to[looseness=20,in=150,out=-150] node[midway,left]{$1$} (z2);

\node[v] (xx) at (1,-1) {};
\node[v,label={70:$a_1$}] (yx) at (2,-1) {};
\node[v,label={70:$a_2$}] (zx) at (3,-1) {};

\node[v] (xy) at (1,-2) {};
\node[v] (yy) at (2,-2) {};
\node[v,label={0:$a_3$}] (zy) at (3,-2) {};

\node[v] (xz) at (1,-3) {};
\node[v] (yz) at (2,-3) {};
\node[v] (zz) at (3,-3) {};

\draw[->] (yy) --node[midway,above]{$0$} (xx);
\draw[->] (yy) --node[midway,above]{$3$} (xz);
\draw[->] (yy) --node[midway,above]{$3$} (zx);
\draw[->] (yy) --node[midway,above]{$6$} (zz);

\draw[->] (yz) --node[midway,below]{$1$} (xz);
\draw[->] (yz) --node[midway,below]{$4$} (zz);

\draw[->] (zy) --node[midway,right]{$1$} (zx);
\draw[->] (zy) --node[midway,right]{$4$} (zz);

\draw[->] (zz)  to[looseness=20,in=-15,out=-60] node[midway,right]{$2$} (zz);

\node at (0, 0) {$\times$};
\node at (-0.1, -4) {$\mathbb{C}$};
\node at (2, -4) {$\mathbb{C}^2$};
\end{scope}
\end{tikzpicture}
\caption{Left: a digraph structure $\mathbb{C}$ (the signature has a single symbol of arity 2) with three vertices and three arcs (tuples) with finite values – the remaining arcs have infinite values and are not drawn.
Right: the product $\mathbb{C}\times \mathbb{C}$ (with three vertices $a_1,a_2,a_3$ distinguished for later).}
\label{fig:product}
\end{figure}

We say that $a \in C$ is \emph{dominated} by $b \in C$
if there is an $M > 0$ such that for all $(f,\vx)\in\tup(\C)$ with $x_i=a$, we have
\[ \fC(x_1,\dots,x_{i-1},b,x_{i+1},\dots) \leq M \cdot \fC(\vx). \]
We say that $a\in C$ is dominated in $\C$ if $a$ is dominated by $b \neq a$ for some $b\in C$.
A sequence of $\sigma$-structures $\C_0,\dots,\C_\ell$ is a \emph{dismantling
sequence} if there exists $a_i \in C_i$
such that $a_i$ is dominated in $\C_i$, and $\C_{i+1}$ is the substructure of
$\C_i$ induced by $C_i\setminus\set{a_i}$,
for $i\in\set{0,\dots,\ell-1}$.
In this case, we say that $\C_0$ dismantles to $\C_\ell$.
A structure $\C$ is \emph{diagonalisable} if $\C^2$ dismantles to the substructure induced by its diagonal  $\Delta(C^2)=\set{(c,c)  \mid c \in C}$.

\begin{example}
Consider $\C^2$ in Figure~\ref{fig:product}. Let $f$ be the unique symbol (of arity two) in the signature. Let $a_1,a_2,a_3$ be the vertices as drawn and let $b$ be the vertex of $\C^2$ with a loop of value 2. Then, for example, the value of the arc from $a_3$ to $b$ is 4, or more formally, $f^{\C}(a_3,b) = 4$. For the vertex $a_1$, all incident arcs have value $\infty$ (formally, $f^{\C}(a_1,x) = f^{\C}(x,a_1) = \infty$ for all $x \in V(\C^2)$), so it is dominated by every other vertex. 
The vertex $a_2$ is dominated by $b$ (with $M=4$).
After removing $a_2$, the vertex $a_3$ is dominated by $b$ as well (this is false before removing $a_2$, because $f^\C(a_3,a_2) =1$ while $f^\C(b,a_2) = \infty$, so we cannot guarantee $f^\C(b,x) \leq M \cdot f^\C(a_3,x)$).
Thus $\C^2, \C^2\setminus\{a_1\}, \C^2\setminus\{a_1,a_2\},\C^2\setminus\{a_1,a_2,a_3\}$ is a dismantling sequence.
Symmetrical vertices can be similarly dominated, hence $\C^2$ dismantles to $\Delta(C^2)$, meaning $\C$ is diagonalisable.
On the other hand, $\C$ is not a MinSol structure (there is no way to order the two peripheral vertices).
We refer to~\cite{BrightwellW00} for more examples of dismantlable graphs.
\end{example}

Homomorphisms $\psi,\phi$ from $\C$ to $\D$ are \emph{adjacent} if there exists $M>0$ such that for all $(f,\vx)\in\tup(\C)$ and $\vy \in D^{\ar(f)}$ with $y_i\in\set{\psi(x_i),\phi(x_i)}$, we have
\begin{equation}\label{eq:adjM}
	\fD(\vy) \leq M \cdot \fC(\vx).
\end{equation}
Thus $a$ is dominated by $b$ in $\C$ if and only if the function $s \colon C \to C\setminus\set{a}$ that maps $a$ to $b$ and everything else identically is a homomorphism from $\C$ to $\C$,
 and $s$ is adjacent to the identity homomorphism.
(This is stronger than just $s$ being a homomorphism, since $f^\C(a,a,a)=0$ implies not only $f^\C(b,b,b) = 0$, but also $f^\C(a,a,b)=0$, for example).
Note that adjacency is a symmetric but  not a transitive property.

Finally, for $\sigma$-structures $\C$ and $\D$, we define the \emph{link graph} $L(\C,\D)$ to be the simple graph whose vertices are the homomorphisms from $\C$ to $\D$, with edges between adjacent homomorphisms.

The following theorem was proved in~\cite[Theorem~3.6]{Briceno21:jctb} for
relational structures but the result easily extends to structures.

\begin{theorem}\label{thm:diag}
  Let $\C$ be a $\sigma$-structure. Then, the following are equivalent.
	\begin{itemize}
		\item $\C$ is diagonalisable;
		\item 	$\pi_1$ and $\pi_2$ are connected in $L(\C^2,\C)$ by a path of
      adjacent idempotent homomorphisms.\\
		(We say a function $\psi: C^2 \to C$ is idempotent if $\psi(c,c)=c$ for all $c \in C$.)
	\end{itemize}
\end{theorem}
\begin{proof}
	This follows from the fact that $\rel{\C^2}=\rel\C^2$
	and that our definitions are the same as those of~\cite[Theorem~3.6]{Briceno21:jctb} applied to $\mathbb H=\C$ and $J=\Delta(C^2)$.
	Specifically a function $\phi: C \to C$ is a homomorphism from $\C$ to $\C$
	if and only if it is a homomorphism from $\rel\C$ to $\rel\C$. 
	Similarly, $a$ is dominated by $b$ in $\C^2$ if and only if
	$a$ is dominated by $b$ in $\rel{\C^2}=\rel\C^2$.
	Thus $\C$ is diagonalisable if and only if $\rel{\C}^2$ dismantles to its full diagonal (not just any subset of it).
	Further $\phi,\psi: C \to C$ are adjacent homomorphisms	from $\C$ to $\C$
	if and only if they are adjacent homomorphisms from $\rel\C$ to $\rel\C$.	
	Finally, $\pi_1, \pi_2$ are connected by a path of adjacent idempotent homomorphisms
	if and only if they are $J$-connected by any homomorphism in $L(\rel{\C}^2, \C)$
	in the sense of~\cite{Briceno21:jctb}.
\end{proof}

We now show that diagonalisability is more general than the $\MinSol$ condition.

\begin{lemma}\label{lem:minsol-dis}
	Let $\C$ be a \MinSol structure. Then $\C$ is diagonalisable.
	Moreover, there is a path on 3 vertices between $\pi_1$ and $\pi_2$ in $L(\C^2,\C)$.
\end{lemma}
\begin{proof}
	Define $\phi: C^2 \to \Delta(C^2)$ by $\phi(x,y)=(\max(x,y), \max(x,y))$,
	where $\max$ is with respect to the total order on $C$.
	We claim for each $(x,y) \in C^2$, $a \coloneqq (x,y)$ is dominated by $b \coloneqq \phi(x,y)$.
	Indeed, for each $(f,(\vx,\vy))\in\tup(\C^2)$ with $(x_i,y_i)=a$ and $n \coloneqq \ar(f) > 1$, we have
	\begin{align*}
	     &f^{\C^2}\left((x_1,y_1),\dots,b,\dots,(x_n,y_n)\right) \\
	   ={} &f^{\C}(x_1,\dots,\max(x_i,y_i),\dots,x_n) + f^{\C}(y_1,\dots,\max(x_i,y_i),\dots,y_n) \\
	   \leq{} & M \cdot f^{\C}(x_1,\dots,x_i,\dots,x_n) + M \cdot f^{\C}(y_1,\dots,y_i,\dots,y_n) \\
	   ={} &M \cdot f^{\C^2}\left((\vx,\vy)\right),
	\end{align*}
	for some $M>0$, where the inequality follows from the assumption that $\C$ is a $\MinSol$ structure.	
	For $f \in \sigma$ with $\ar(f) = 1$, we have that $a$ is dominated by $b$ because with $M \geq2$ we always have
	\[f^{\C^2}(b) = f^{\C}(\max(x,y)) + f^{\C}(\max(x,y)) \leq M \cdot (\fC(x) + \fC(y)) = M \cdot f^{\C^2}(a).\]

	Therefore, we can dismantle the non-diagonal elements $(x,y)$ in any order 
	to obtain a dismantling sequence from $\C^2$ to the substructure induced by $\Delta(C^2)$.

	Let $\mu: C^2 \to C$ be defined by $\mu(x,y)=\max(x,y)$, where $\max$ is with respect to the total order defined on $C$.
	Then similarly as above, one can check $\pi_1, \mu, \pi_2$ is a path in $L(\C^2,\C)$.
\end{proof}

We remark that~\cite{Briceno21:jctb} show many other equivalent formulations,
including a property known as \emph{finite duality}.
They also discuss how finite duality allows to efficiently solve many problems such as homomorphism extensions.
However, in our setting it is $\rel\C$ rather than $\C$ that is restricted, so such a property would not take finite, positive values of $\C$ into account.

Instead, our approach is based on Baker's technique:
we partition graph into breadth-first-search layers and
use the fact that the problem can be solved exactly on a subgraph induced by a few consecutive layers.
To merge such solutions into one, we use a small number of overlapping layers
and use the path between projections $\pi_1,\pi_2$ given by \cref{thm:diag} to ``blend in'' two solutions.
By increasing the number of exactly solved, non-overlapping layers, we can reduce any loss due to differences between finite, positive values.

\subsection{PTAS}
Baker's approach relies on the following structural property of planar graphs,
which is e.g.\ a direct consequence of~\cite[Theorem 83]{bodlaender98:tw}.
\begin{lemma}\label{lem:planar-tw}
Let $G$ be a planar graph and $v_0 \in V(G)$ be an arbitrary vertex.
Let $L_i$ be the set of vertices at distance exactly $i$ from $v_0$ (i.e.\ the $i$th layer of a BFS from $v_0$). Then, the subgraph induced by any $t$ consecutive layers $G[L_i \cup L_{i+1} \cup \dots \cup L_{i+t+1}]$ has treewidth at most $3t$.
\end{lemma}

\begin{theorem}\label{thm:minhom}
Let $\cP$ be the class of planar graphs.
Then, for any $\eps>0$ and any VCSP instance $(\A,\C)$ with
$\Gaifman(\A)\in\cP$ and $\C$ diagonalisable, we can find a solution of value
at most $(1+\eps) \minval(\A,\C)$ in time $\lvert\A\rvert \cdot c^{1/\eps}$
where $c$ depends on $\C$ only.
\end{theorem}
\begin{proof}
	Let $(\A,\C)$ be a VCSP instance as per the theorem.
	Generally, for any left-hand side structure $\B$, we will write $\val_{\B}(h)$ for the value of an assignment $h:B \to C$ with respect to the instance $(\B,\C)$, and write $\val(\cdot)$ for $\val_\A(\cdot)$ by default.
	
  	By \cref{thm:diag} there is a sequence of adjacent
    homomorphisms $\psi_1, \dots, \psi_\ell$ from $\C^2$ to $\C$ such that
    $\psi_1=\pi_1$ and $\psi_\ell=\pi_2$.
	Let $M \geq 1$ be sufficiently large such that~\cref{eq:adjM} holds for all
  adjacent homomorphisms $\psi_i$ and $\psi_{i+1}$, $i\in\set{1,\dots,\ell-1}$.
	Let $k \coloneqq \lceil \frac{2M}\eps \rceil$.

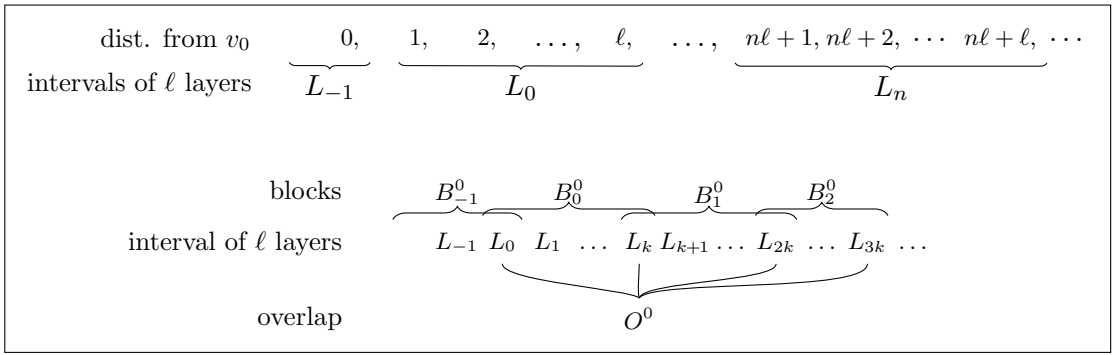
\begin{figure}[b!]
\begin{tikzpicture}[framed]
\begin{scope}[xscale=0.9]
\node[anchor=east] at (-1.3,0) {\small{dist. from $v_0$}};
\node[anchor=east] at (-1.3,-0.6) {\small{intervals of $\ell$ layers}};
\node (a0) at (0,0) {\footnotesize{$0,$}};
\node (a1) at (1,0) {\footnotesize{$1,$}};
\node (a2) at (2,0) {\footnotesize{$2,$}};
\node (a3) at (3,0) {$\vphantom{1}\dots,$};
\node (a4) at (4,0) {\footnotesize{$\ell,$}};
\node (a6) at (5,0) {$\vphantom{1}\ldots,$};
\node (a8) at (6.5,0) {\footnotesize{$n\ell+1,\quad$}};
\node (a9) at (7.5,0) {\footnotesize{$n\ell+2,$}};
\node (a10) at (8.5,0) {$\dots$};
\node (a11) at (9.5,0) {\footnotesize{$n\ell+\ell,$}};
\node (a12) at (10.5,0) {$\dots$};
\draw[decorate, decoration ={brace,mirror,raise=1pt}] ($(a0.south west)-(0.6,0)$) -- (a0.south east)
            node[midway, below=2pt] {$L_{-1}$};
\draw[decorate, decoration ={brace,mirror,raise=1pt}] (a1.south west) -- (a4.south east)
            node[midway, below=2pt] {$L_0$};            
\draw[decorate, decoration ={brace,mirror,raise=1pt}] (a8.south west) -- (a11.south east)
            node[midway, below=2pt] {$L_n$};                        
\end{scope}            
            
\begin{scope}[shift={(2,-2.7)}, xscale=0.6]
\node[anchor=east] at (-3.3,0.7) {\small{blocks}};
\node[anchor=east] at (-3.3,0) {\small{interval of $\ell$ layers}};
\node[anchor=east] at (-3.3,-1) {\small{overlap}};

\node (aa1) at (-1,0) {\footnotesize{$L_{-1}$}};
\node (a0) at (0,0) {\footnotesize{$\vphantom{L_{-1}}L_0$}};
\node (a1) at (1,0) {\footnotesize{$\vphantom{L_{-1}}L_1$}};
\node (a2) at (2,0) {\footnotesize{$\vphantom{L_{-1}}\dots$}};
\node (ak) at (3,0) {\footnotesize{$\vphantom{L_{-1}}L_k$}};
\node (akp) at (4,0) {\footnotesize{$\vphantom{L_{-1}}L_{k\hspace*{-0.5pt}+\hspace*{-0.7pt}1}$}};
\node (a5) at (5,0) {\footnotesize{$\vphantom{L_{-1}}\dots$}};
\node (a2k) at (6,0) {\footnotesize{$\vphantom{L_{-1}}L_{2k}$}};
\node (a7) at (7,0) {\footnotesize{$\vphantom{L_{-1}}\dots$}};
\node (a3k) at (8,0) {\footnotesize{$\vphantom{L_{-1}}L_{3k}$}};
\node (a9) at (9,0) {\footnotesize{$\vphantom{L_{-1}}\dots$}};
\draw[decorate, decoration ={brace,raise=1pt,amplitude=5pt}] ($(aa1.north west)-(0.7,0)$) -- ($(a0.north east)-(0.1,0)$)
            node[midway, above=2pt] {\footnotesize{$B^0_{-1}$}};
\draw[decorate, decoration ={brace,raise=2.5pt,amplitude=5pt}]  ($(a0.north west)+(0.1,0)$) -- ($(ak.north east)-(0.2,0)$)
            node[midway, above=3pt] {\footnotesize{$B^0_{0}$}};            
\draw[decorate, decoration ={brace,raise=1pt,amplitude=5pt}]  ($(ak.north west)+(0.15,0)$) -- ($(a2k.north east)-(0.2,0)$)
            node[midway, above=2pt] {\footnotesize{$B^0_{1}$}};                        
\draw[decorate, decoration ={brace,raise=2.5pt,amplitude=5pt}]  ($(a2k.north west)+(0.2,0)$) -- ($(a3k.north east)-(0.2,0)$)
            node[midway, above=3pt] {\footnotesize{$B^0_{2}$}};                   
            
\node (o) at (3,-1) {\footnotesize{$O^0$}};
\draw (o.north) to[looseness=0.2,out=150,in=-90] (a0);
\draw (o.north) to[looseness=0.5,out=120,in=-90] (ak);
\draw (o.north) to[looseness=0.2,out=70,in=-90] (a2k);
\draw (o.north) to[looseness=0.2,out=40,in=-90] (a3k);
\end{scope}            
            
\end{tikzpicture}
\caption{Illustration of the sets $L_n$, $B^i_j$ and $O^i$ for $i=0$ (other $i$ look the same, just shifted).}
\label{fig:overlap}
\end{figure}		

	Let $\A$ be a $\QQ_{\geq 0}$-valued structure and let $G=\Gaifman(\A) \in \cP$ be its Gaifman graph.
	Fix an arbitrary vertex $v_0 \in A$ in $\Gaifman(\A)$.
	For $n \in \ZZ$, let $L_n \subseteq A$ be the set of vertices whose distance from $v_0$ is in $\{n \ell + 1, \dots, n \ell + \ell\}$.
	So $L_n$ are intervals of $\ell$ layers, which partition the vertex set $A$.
	See Figure~\ref{fig:overlap}.
	For each $j\in\ZZ$ and $i \in [k]$ let
	\[
		B_j^i \coloneqq L_{jk-i} \cup \dots \cup L_{jk-i+k},
	\]
	so that $B_j^i$ is a block of $(k+1) \cdot \ell$ consecutive layers.
	Iterating through the indices $j$ gives consecutive blocks that overlap on $\ell$ layers;
	the index $i$ shifts which layers are in the overlap.
	That is,
	\[
		B_j^i \cap B_{j+1}^i = L_{(j+1)k-i}.
	\]
	Define the overlaps $O^i = \bigcup_{j} B_j^i \cap B_{j+1}^i$ for $i \in [k]$.
	We note that the $O^1,\dots,O^k$ are disjoint.

	Consider an optimal solution $h^*: A \to C$ for the VCSP instance $(\A,\C)$.
	As the $O^i$ are all disjoint, there exists $\istar\in[k]$ with%
	\[ \val_{\A[O^\istar]}(\restr{h^*}{O^\istar}) \leq \frac{1}{k} \val(h^*) \leq \frac\eps{2M} \val(h^*). \]
	We henceforth write $B_j=B_j^\istar$ and $O=O^\istar$.
	Note that, just as in Baker's original approach, the choice of $\istar$ is not available to the algorithm, as we do not know $h^*$.
	However, as the number of choices for $i\in[k]$ is linear in $1/\eps$, we can proceed with each possible $i$, construct the solution $h'$ as discussed below and output the one with the lowest value $\val(h')$.
		
	Let $\A^+$ be a $\sigma$-structure with domain $A$ defined by
	\[
		f^{\A^+}(\vx) = \begin{cases}
			M \cdot \fA(\vx) &\text{if $\toset{\vx} \subseteq O$} \\
			\fA(\vx) &\text{otherwise},
		\end{cases}
	\]
	so that tuples which lie within $O$ are amplified by a factor of $M$.
	
	For each $j$, the Gaifman graph $\Gaifman(\A^+[B_j])$ has treewidth at most $\BigO{(k+1)\ell} = \BigO{M \ell /\eps}$ by \cref{lem:planar-tw}.
  Thus for each $j$, we can find a tree decomposition~\cite{Bodlaender93:linear-tw} and compute an optimal solution $h_j$ to $(\A^+[B_j],\C)$
	in total time $\lvert\A\rvert \cdot {\lvert\C\rvert}^{\BigO{M \ell/\eps}}$
  via standard dynamic programming~\cite{reed2003algorithmic}. 
	Then by optimality of $h_j$,
	\[ \val_{\A^+[B_j]}(h_j) \leq \val_{\A^+[B_j]}(\restr{h^*}{B_j}). \]
	Therefore, summing over all $j$,
	we count the contribution of every constraint once,
	except for constraints whose scope is contained in $O$
	(and thus in exactly two sets $B_j$),
	which are counted $2M$ times in total:
	\begin{equation}
	\begin{split}\label{eq:sum_val_j}
		\sum_{j} \val_{\A^+[B_j]}(h_j)
		&\ \leq\ \sum_{j} \val_{\A^+[B_j]}(\restr{h^*}{B_j}) \\
		&\ =\ \val_\A(h^*) + (2M-1) \cdot \val_{\A[O]}(\restr{h^*}{O})
		 \ \leq\ (1+\eps) \val(h^*).
	\end{split}
	\end{equation}

	Observe that for each $x \in A$, either $x\not\in O$ and there is a unique~$j$ for which $x \in
  B_j$, or $x \in O$ and there is a unique $j$ for which $x \in B_j \cap B_{j+1}$.
  In the latter case, $x \in L_{(j+1)k-\istar}$ and we let $s \in [\ell]$ denote the unique~$s$ for which $x$
  is at distance exactly $((j+1)k-\istar)\ell + s$ from $v_0$.
	Let $h': A \to C$ be defined as follows
	\[
		h'(x) = \begin{cases}
			h_j(x) &\qquad \text{if $x \in B_j$ for a unique $j$} \\
			\psi_s\big(h_j(x),h_{j+1}(x)\big) &\qquad \text{if $x \in B_j \cap B_{j+1}$ and $d(x,v_0)=((j+1)k-\istar)\ell + s$.}
		\end{cases}
	\]
	We claim that $h'$ is a solution to $(\A,\C)$ with $\val(h') \leq (1+\eps)\minval(\A,\C)$.
	Let $(f,\vx) \in \tup_{>0}(\A)$.
	Note that by definition of the Gaifman graph $\Gaifman(\A)$, all $x_i$ are adjacent to each other, 
	so $\toset\vx$ is contained in one or two consecutive layers.
	Consider the following two cases.
	\begin{enumerate}
		\item If $\toset\vx \not\subseteq O$, then there is a unique $j$ such that
      $\toset{\vx} \subseteq B_j$, and so $h'(x_i)=h_j(x_i)$ for each $i$,
      because either:
      $x_i \not\in O$ and so $h'(x_i)=h_j(x_i)$ by definition of $h'$,
      or $x_i$ is in the last layer of $B_{j-1} \cap B_j$ and      
      $h'(x) = \psi_\ell\big(h_{j-1}(x),h_j(x)\big) = \pi_2\big(h_{j-1}(x),h_j(x)\big) = h_j(x)$,
      or analogously $x_i$ is in the first layer of $B_j \cap B_{j+1}$ and $\psi_1=\pi_1$.
		Thus
		\begin{equation}\label{eq:hcase1}
			\fC(h'(\vx))=\fC(h_j(\vx)).
		\end{equation}
		\item Else, if $\toset\vx \subseteq O$, then there is a unique $j$ such that
		$\toset{\vx} \subseteq B_j \cap B_{j+1} = L_{(j+1)k-\istar}$.
		Since $\toset\vx$ is contained in two consecutive layers, there is some $s$ such that
		all vertices in $\toset{\vx}$ are at distance $((j+1)k-\istar)\ell+s$ or $((j+1)k-\istar)\ell+s+1$ from $v_0$.
		Thus
		\begin{equation*}
			h(x_i) \in \left\{\psi_s    \big(h_j(x_i),h_{j+1}(x_i)\big),\, 
			                  \psi_{s+1}\big(h_j(x_i),h_{j+1}(x_i)\big)\right\}
		\end{equation*}
		for each $x_i$. Finally, as $\psi_s$ and $\psi_{s+1}$ are adjacent
		\begin{equation}\label{eq:hcase2}
			\fC(h'(\vx)) \leq M \cdot f^{\C^2}\big(h_j(\vx),h_{j+1}(\vx)\big) = M \cdot \left(\fC\big(h_j(\vx))+\fC(h_{j+1}(\vx)\big)\right).
		\end{equation}
	\end{enumerate}
	Thus, by~\cref{eq:hcase1,eq:hcase2,eq:sum_val_j},
	\begin{align*}
		\val_{\A}(h') \leq \sum_{j} \val_{\A^+[B_j]}(h_j) \leq (1+\eps)\val(h^*) = (1+\eps)\minval(\A,\C),
	\end{align*}
	and so $h'$ is the solution we seek.
\end{proof}

\begin{remark}\label{rmk:planar}
	In \cref{thm:minhom} it would be sufficient to require that $\C^2$ dismantles to any substructure of its diagonal, as opposed to its full diagonal (as in the definition of diagonalisability).
	By~\cite[Theorem~3.6]{Briceno21:jctb} (extended as in \cref{thm:diag}) this is equivalent to saying that $\C$ dismantles to a substructure $\bbI$ such that $\bbI$ is diagonalisable.
	
	In this case, $\pi_1$ and $\pi_2$ are still connected in $L(\C^2,\C)$, but the homomorphisms in the path connecting them will not be necessarily idempotent.
	However, the above proof (for the case of planar graphs) did not rely on this property.
	This is in contrast with \cref{thm:minhomBaker} (for Baker classes) where we actually use the fact that the homomorphisms are idempotent.
\end{remark}

Since a $\MinSol$ structure $\C$ is diagonalisable by \cref{lem:minsol-dis},
we have the following corollary.

\begin{corollary}
Let $\cP$ be the class of planar graphs.
Given any $\eps>0$ and instance $(\A,\C)$ of $\MinSol_{\cP}$,
we can find a solution of value at most $(1+\eps)\minval(\A,\C)$ in time $\lvert\A\rvert \cdot c^{1/\eps}$, where $c$ depends on $\C$ only.
\end{corollary}

We remark the proof yields $c^{1/\eps} = {\lvert\C\rvert}^{\BigO{M\ell/\eps}}$,
and for $\MinSol$ structures \cref{lem:minsol-dis} yields $\ell=3$;
hence when the bound $M$ is a constant (e.g.\ for $\{0,1,\infty\}$-valued $\MinSol$ structures)
the dependency on $\C$ is simply ${\lvert\C\rvert}^{\BigO{1/\eps}}$.

\section{Minimisation on Baker classes}\label{sec:baker}
\subsection{Definition of Baker classes}
A \emph{layering} of a graph $G$ is a function $\lambda \colon V(G) \to \ZZ$ such that $|\lambda(u) - \lambda(v)| \leq 1$ for adjacent vertices $u,v$ in $G$.
That is, vertices of $G$ are partitioned into \emph{layers} $\lambda^{-1}(i)$ for $i \in \ZZ$ and edges only go within one layer or between two consecutive layers.

Baker's technique~\cite{Baker94:jacm} relies on a layering of planar graphs such
that the subgraph induced by $\lambda^{-1}(I)$, for any interval $I$ (a set of a
few consecutive integers), has bounded treewidth (the bound depending only on
$|I|$), as formally stated in \cref{lem:planar-tw}.
As one might imagine, this can be iterated: it would suffice that the subgraph induced by $\lambda^{-1}(I)$ itself has such a ``bounded treewidth layering''.
Consider now the class of graphs obtained from planar graphs by adding a single
vertex, adjacent to all the others; then a layering can only have three
non-empty layers;\footnote{There can be a non-empty layer before the universal
vertex, a layer containing the universal vertex, and a layer after.}
nevertheless, an algorithm can easily circumvent this by guessing the assignment to the single new vertex (i.e.\ iterating through all possibilities).

Dvo\v{r}\'{a}k~\cite{Dvorak20:soda} defined a \emph{Baker class} as any class of graphs that can be dealt with in the above ways.
Informally, a class of graphs is Baker if any graph in the class can be reduced to an empty graph by a bounded number of operations: either removing a single vertex, or selecting some layering $\lambda$ and continuing separately with every interval in that layering (each subgraph induced by $\lambda^{-1}(I)$ for intervals $I$ of at most some size).
It turns out the notion of treewidth is not necessary here, as graphs of bounded treewidth also form a Baker class.

Before we state the definition formally, let us make a few remarks.
The idea of iteratively going through ``some layering'' and then into ``every interval'' is conveniently phrased as a strategy winning a game in a \emph{bounded} number of rounds (this will be particularly useful when we will want to state an assumption that the choices, including layerings, can be constructed efficiently).
Dvo\v{r}\'{a}k's definition considers graphs with a total ordering of their vertex set. Roughly speaking, this is to restrict the definition to ``monotone'' strategies, where the single vertices to be deleted are decided upfront --- this restriction won't be important for us, but we will state it as in~\cite{Dvorak20:soda} (let us also remark that layerings are not restricted by the ordering).
Finally, the definition is made a bit complicated by the fact that the number of consecutive layers we may need to include in an interval may depend on how deep we go (how many iterations of the game are done). 
This dependency is formalised as a function $r$ below: for a fixed problem and approximation ratio the reader should think of some arbitrarily quickly increasing function $r \colon \NN \to \NN$.
We now proceed with the formal definition.\looseness=-1

\begin{definition}
For a graph $G$ and a function $r \colon \NN \to \NN$, the \emph{Baker game} on $(G,r)$ is defined as follows, for two players I and II\@.
Player I starts by selecting a total ordering of $V(G)$.
A state of the game is a pair $(G',t)$ where $G'$ is an induced subgraph of $G$ (with its ordering inherited from $G$) and $t$ is an integer describing how many rounds have passed.
The initial state is $(G,0)$ and Player I wins in $t$ rounds if the state $(\emptyset,t)$ is reached, for any $t$.
Otherwise, in state $(G',t)$, Player I chooses one of the following actions:
\begin{itemize}
\item delete the first vertex $v$ of $G'$, according to the ordering; Player II then takes no action and the game continues in state $(G'-v, t+1)$;
\item select a layering $\lambda$ of $G'$; Player II then selects an interval $I$ of at most $r(t)$ (and no more than $|V(G)|$) consecutive integers and the game continues in state $(G'[\lambda^{-1}(I)], t+1)$.
\end{itemize}
\end{definition}

\noindent
We say a function $f \colon \NN \to \NN$ is \emph{sub-additive} if $f(x) + f(y) \leq f(x+y)$.

\begin{definition}
  For sub-additive\footnote{Sub-additivity is satisfied by any reasonable
  time-complexity bound function and is used implicitly in~\cite{Dvorak20:soda}.} functions $f_1,f_2 \colon \NN \to \NN$, we say
	a class of graphs $\cG$ is \emph{($f_1,f_2$)-efficiently Baker}
	if for every function $r \colon \NN \to \NN$ there exists an integer~$t_{\max}$
	and an algorithm such that:
	for all $G \in \cG$, the algorithm wins as Player I in the Baker game on $(G,r)$ in at most $t_{\max}$ rounds,
	using time $f_1(\lvert G\rvert)$ to compute the initial ordering
	and using time $f_2(\lvert G\rvert)$ to determine the action at each state of the game.
\end{definition}

As discussed in the introduction, efficiently Baker classes generalise
excluded-minor classes.

\begin{theorem}[Dvo\v{r}\'{a}k~{\cite[Theorem 2.1]{Dvorak20:soda}}]
	Let $\cG$ be a class of graphs that excludes a minor.
	Then $\cG$ is $(\BigO{n^2},\BigO{n})$-efficiently Baker.
\end{theorem}

\subsection{PTAS}
The proof for Baker classes largely follows the proof for planar graphs.
To main difference is that in order to handle exceptional vertices (also known as ``apex'' vertices) we will need to guess and fix a partial assignment $\rho$ on them.
When ``blending in'' two assignments $h_j(x)$ and $h_{j+1}(x)$ to some variable $x$ into a single assignment $\psi_s\big(h_j(x),h_{j+1}(x)\big)$,
we will need to preserve the partial assignment~$\rho$, if it is defined on $x$.
This is why we need the homomorphisms $\psi_s : \C^2 \to \C$ to be idempotent, so that $\psi_s(\rho(x),\rho(x)) = \rho(x)$.

\begin{theorem*}[\Cref{thm:minhomBaker} restated]
\StateMinSolBaker%
\end{theorem*}
\begin{proof}
	For a structure $\A$ and an induced subgraph $G'$ of $\Gaifman(\A)$, we write $\A[G']$ as a shorthand for the induced substructure $\A[V(G')]$ of $\A$.

	Since $\C$ is diagonalisable, there is a sequence $\psi_1,\dots,\psi_\ell$ of
  adjacent idempotent homomorphisms $\C \to \C$ from $\psi_1=\pi_1$ to $\psi_\ell=\pi_2$.
	Let $M$ be sufficiently large such that~\cref{eq:adjM} holds for all adjacent homomorphisms $\psi_i$ and $\psi_{i+1}$.
	For $\eps > 0$, define $r(t) \coloneqq 2M\ell \cdot \lceil\frac{1}{\eps}\rceil \cdot 2^t$.
	Let $\Alg$ be the algorithm certifying that $\cG$ is efficiently Baker and let $t_{\max}$ be the integer that is guaranteed to exists for $r$.
	Our algorithm starts by using $\Alg$ to compute the vertex ordering of $\Gaifman(\A)$.
	
	We then proceed with a recursive procedure.
	The input of the procedure consists of an instance $\A'$, a partial assignment
  $\rho$ from at most $t_{\max}$ elements $\dom(\rho)$ in $\A'$ to $\C$, and a
  state $(G',t)$ of the Baker game where $G'$ is equal to the Gaifman graph of
  $\A' \setminus \dom(\rho)$.
	We describe the procedure and claim inductively that it computes a solution of value at most $e^{\eps/2^t}$ times the optimum (among all total assignments that agree with $\rho$ on $\dom(\rho)$).
	Moreover, we claim the procedure finishes in time at most ${(2 |C| \cdot M \cdot \lceil \frac{1}{\eps} \rceil \cdot 2^{t_{\max}})}^{t_{\max}} \cdot f_2(\lvert G'\rvert \cdot 2^{t_{\max}})$.
	Starting the recursive procedure with the instance $\A$, the state $(\Gaifman(\A),0)$, and empty $\rho$, this will conclude the proof.\looseness=-1
	
	Let $\val(\rho)$ be the value of constraints fully contained in $\dom(\rho)$:
	\[\val(\rho) \coloneqq \sum_{\substack{(f,\vx) \in \tup(\A) \\ \toset{\vx}\subseteq \dom(\rho)}} \fA(\vx) \fC(\rho(\vx)).\]	
	Since all the assignment we consider will agree with $\rho$, $\val(\rho)$ will be a common part of all of them.
	We will not include it in the approximation we inductively claim to get;
	that is, we claim the recursive procedure will output an assignment $h$ from $\A'$ to $\C$ such that
	\[\val^\rho(h) \leq e^{\eps/2^t} \cdot \val^\rho(h^*),\]
	where $\val^\rho(h)$ is a shorthand notation for $\val(h) - \val(\rho)$ and $h^*$ is an optimal solution, among solutions that agree with $\rho$ on $\dom(\rho)$.

	In state $(\emptyset, t)$ for any $t$, our algorithm outputs the assignment
  $\rho$, which is a total assignment since $\A' \setminus \dom(\rho)$ is empty.
	It is trivially optimal (among assignments that agree with $\rho$).
	
	In state $(G',t)$, if $\Alg$ selects to delete the minimum vertex $v$ of $G'$, we consider each possible assignment to $v$ in $C$, and recursively call the procedure with the state $(G'-v,t+1)$, the instance $\A'$, and the assignment $\rho$ extended to $v$.
	We output the best solution found this way: since we consider all possible assignments, we get a solution as close to optimum as guaranteed inductively for $t+1$.
	In the final running time, the branching over all assignments to $v$ in $C$ will contribute at most a factor of $|C|^{t_{\max}}$.
	
	If $\Alg$ selects a layering $\lambda$ of $G'$, we consider each possible shift index $i \in [k]$, where $k \coloneqq (2M-1) \cdot \lceil \frac{1}{\eps} \rceil \cdot 2^t$.
	This branching will contribute to the final running time a factor of at most ${(2M \cdot \lceil \frac{1}{\eps} \rceil \cdot 2^{t_{\max}})}^{t_{\max}}$.
	For $n \in \ZZ$, let
		\[L_n \coloneqq \{n\ell + 1, \dots, n \ell + \ell\},\]
	so $L_n$ are intervals of length $\ell$ which partition $\ZZ$;
		\[B^i_j \coloneqq L_{jk-i} \cup \dots \cup L_{jk-i+k},\]
	so $B^i_j$ consists of $k+1$ such intervals, of which the last one overlaps with $B^i_{j+1}$;
		\[O^i \coloneqq \bigcup_{j} B^i_j \cap B^i_{j+1} = \bigcup_j L_{(j+1)k-i},\]
	so $O^i$ are disjoint for different $i \in [k]$.
	Note that we chose $r(t)$ to satisfy $|B_j| = (k+1) \ell \leq r(t)$.	

	For a subset of integers $I$ (such as $B^i_j$ or $O^i$) we henceforth abuse notation and write $\A'[I]$ as a shorthand for $\A'[\lambda^{-1}(I) \cup \dom(\rho)]$ and $\restr{h}{I}$ as a shorthand for $\restr{h}{\lambda^{-1}(I) \cup \dom(\rho)}$.
	
	Let $h^*$ be an optimal solution to $(\A',\C)$ that agrees with $\rho$.
	Since $O^i$ are disjoint for different $i \in [k]$, so are the constraints of $\A'[O^i]$ (i.e.\ those contained in $\lambda^{-1}(O^i) \cup \dom(\rho)$), except for those fully contained in $\dom(\rho)$ (and accounted for in $\val(\rho)$).
	This implies there exists a shift index $\istar \in [k]$ such that
	\[ \val^\rho_{\A'[O^\istar]}(\restr{h^*}{O^\istar})\ \leq\ \frac{1}{k} \val^\rho_{\A'}(h^*). \]
	We henceforth only consider the recursion branch where this holds and skip the superscript~$\istar$, and write $B_j=B_j^\istar$ and $O=O^\istar$.
	
	Let $\A^+$ be a $\sigma$-structure with the same domain as $\A'$, defined by
	\begin{equation*}
		f^{\A^+}(\vx) = \begin{cases}
			M \cdot f^{\A'}(\vx) &\text{if $\toset{\vx} \subseteq \lambda^{-1}(O) \cup \dom(\rho)$} \\
			f^{\A'}(\vx) &\text{otherwise},
		\end{cases}
	\end{equation*}
	so that tuples which lie within $\lambda^{-1}(O) \cup \dom(\rho)$ are amplified by a factor of $M$.
	
	For each $j$, we recurse into state $(G'[\lambda^{-1}(B_j)], t+1)$ computing solutions $h_j$ to $(\A^+[B_j],\C)$ which agree with $\rho$ and which by inductive assumption are almost optimal:
	\[ \val^\rho_{\A^+[B_j]}(h_j)
	   \leq e^{\eps/2^{t+1}} \cdot 
	   \val^\rho_{\A^+[B_j]}(\restr{h^*}{B_j}). \]
	Therefore, summing over $j$ we get (by observing that every constraint is either fully contained in $\dom(\rho)$, or contained in $\A^+[B_j]$ for exactly one $j$, or contained in $\A^+[O]$)
	\begin{align*}
		\sum_{j} \val^\rho_{\A^+[B_j]}(h_j)
		&\leq e^{\eps/2^{t+1}} \cdot \sum_{j} \val^\rho_{\A^+[B_j]}(\restr{h^*}{B_j}) \\		
		&= e^{\eps/2^{t+1}} \cdot \left(\val^\rho_{\A'}(h^*) + (2M-1) \cdot \val^\rho_{\A'[O]}(\restr{h^*}{O}) \right)\\
		&\leq e^{\eps/2^{t+1}} \cdot (1+\frac{2M-1}{k}) \val^\rho_{\A'}(h^*) \\
		&\leq e^{\eps/2^t} \cdot \val^\rho_{\A'}(h^*).
	\end{align*}	
	(the last inequality holds because we chose $k$ to satisfy $1+\frac{2M-1}{k} \leq 1+\frac{\eps}{2^{t+1}} \leq e^{\eps/2^{t+1}}$).
	
	\pagebreak[3]	

	Observe that for each $x \in A'$
	either $\lambda(x) \not \in O$ there is a unique $j$ for which $\lambda(x) \in  B_j$,
	or $\lambda(x) \in O$ and there is a unique $j$ for which $\lambda(x) \in B_j \cap B_{j+1} = L_{(j+1)k-\istar}$.
	In the latter case we let $s(x)$ denote the unique $s \in \{1,\dots,\ell\}$ for which $\lambda(x) = ((j+1)k-\istar) \ell+s$.
	Let $h'\colon A' \to C$ be defined as follows
	\begin{equation*}
		h'(x) \coloneqq \begin{cases}
			\rho(x) &\qquad \text{if $x \in \dom(\rho)$} \\
			h_j(x) &\qquad \text{if $\lambda(x) \in B_j$ for a unique $j$} \\
			\psi_{s(x)}(h_j(x),h_{j+1}(x)) &\qquad \text{if $\lambda(x) \in B_j \cap B_{j+1}$.}
		\end{cases}
	\end{equation*}	

	We claim that $h'$ is a solution to $(\A',\C)$ satisfying $\val^\rho(h') \leq \sum_{j} \val^\rho_{\A^+[B_j]}(h_j) $.
	This will imply
	\[\val^\rho_{\A'}(h') \leq \sum_{j} \val^\rho_{\A^+[B_j]}(h_j) \leq e^{\eps/2^t}\val(h^*),\]
	concluding that $h'$ it is the solution we seek.
	
	Let $(f,\vx) \in \tup_{>0}(\A')$.
	Note that since $\lambda$ is a layering of $G' = \Gaifman(\A') \setminus \dom(\rho)$, there are two consecutive levels which contain all $x_i \in \toset{\vx} \setminus \dom(\rho)$.
	Consider the three cases.
	\begin{enumerate}
		\item If $\toset\vx \subseteq \dom(\rho)$, then for all $j$,
			\[ \fC(h'(\vx)) = \fC(\rho(\vx)) = \fC(h_j(\vx)). \] 
		\item Otherwise, if $\toset\vx \not\subseteq \lambda^{-1}(O) \cup \dom(\rho)$, then there is a unique $j$ such that
      $\toset{\vx} \subseteq \lambda^{-1}(B_j) \cup \dom(\rho)$, and so $h'(x_i)=h_j(x_i)$ for each $i$ (as some $x_i$ might be in the first or the last layer of an overlap, but for those layers we have $\psi_1=\pi_1$ and $\psi_\ell=\pi_2$).
		Thus
		\begin{equation*}
			\fC(h'(\vx))=\fC(h_j(\vx)).
		\end{equation*}
		\item Else, if $\toset\vx \subseteq \lambda^{-1}(O) \cup \dom(\rho)$ (but $\toset\vx \not\subseteq \dom(\rho)$), then there is a unique $j$ such that
      \[\toset{\vx} \subseteq \lambda^{-1}(B_j \cap B_{j+1}) \cup \dom(\rho)\] and there is some $s$ such that
      \[\toset{\vx} \setminus \dom(\rho) \subseteq \lambda^{-1}\left(\big\{((j+1)k-\istar) \ell+s, ((j+1)k-\istar) \ell+s+1\big\}\right).\]
		Thus
		\begin{equation*}
			h'(x_i) \in \set{\psi_s(h_j(x_i),h_{j+1}(x_i)),\, \psi_{s+1}(h_j(x_i),h_{j+1}(x_i))}
		\end{equation*}
		for each $x_i \in \toset{\vx} \setminus \dom(\rho)$.
		Moreover, since $\psi_s$ is idempotent, we can also write for $x_i \in \dom(\rho)$ that 
		\begin{equation*}
			h'(x_i) = \rho(x_i) = \psi_s(\rho(x_i), \rho(x_i)) = \psi_s(h_j(x_i),h_{j+1}(x_i)).
		\end{equation*}

		Therefore, as $\psi_t$ and $\psi_{t+1}$ are adjacent, we have by definition of adjacency that
		\begin{equation*}
			\fC(h'(\vx)) \leq M \cdot f^{\C^2}(h_j(\vx),h_{j+1}(\vx)) = M \cdot (\fC(h_j(\vx))+\fC(h_{j+1}(\vx))).
		\end{equation*}
	\end{enumerate}	
	This concludes the proof that $\val^\rho_{\A'}(h') \leq \sum_{j} \val^\rho_{\A^+[B_j]}(h_j)$ and hence $h'$ is the solution we seek.
	
	To check the running time, observe that every vertex of $G'$ is contained in $\lambda^{-1}(B_j)$ for at most two $j$.
	Hence the total size of graphs $G'[\lambda^{-1}(B_j)]$ is at most $2\lvert G'\rvert$.
	Since $f_2$ is sub-additive, the total time required to consider those graphs in this recursive call and all sub-calls contributes a factor of at most $f_2(2^{t_{\max}}\lvert G'\rvert)$. 
\end{proof}

\section{Maximisation}\label{sec:maxsol}
To present our algorithm for maximisation, we first define what it means for two left-hand side structures $\A,\B$ to be ``close'', in a sense relevant to approximately solving $\MaxSol$.
We then show that there is a dual view which allows to certify ``closeness'' by a fairly concrete mapping: a distribution of partial homomorphisms.
This is then used to show that values given by Sherali-Adams linear programming relaxations of $\MaxSol$ instances on $\A$ and on $\B$ are also close.
Since the level-$k$ Sherali-Adams relaxation solves the problem exactly on instances of treewidth $\BigO{k}$,
it gives a PTAS for classes of structures that are ``close'' to bounded treewidth, as formalised by the notion of ``strong pliability'' below.
The proofs are similar to those in~\cite{rwz21:soda}; the main new contribution is finding a suitable ``dual'' definition (a distribution of partial homomorphisms) that makes the proofs work in the $\MaxSol$ setting. 
We remark we were unable to find an analogue for the $\MinSol$ setting.

\subsection{Pliability}
\begin{definition}\label{def:overcast}
	For two left-hand side $\sigma$-structures $\A, \B$,
	we say $\A$ \emph{strongly overcasts} $\B$, denoted $\A \overcasts \B$,
	if for all \MaxSol $\sigma$-structures $\C$, $\maxval(\A, \C) \geq \maxval(\B, \C)$.
\end{definition}

In contrast,~\cite{rwz21:soda} defined (weak) overcasting in terms of $\Qpos$-valued structures $\C$ only,
instead of the wider class of \MaxSol structures.
The ``strong'' qualifier is only to avoid confusion with~\cite{rwz21:soda}:
we will not consider weak overcasts in this paper, nor analogous weak variants of the definitions given below.

\begin{definition}\label{def:dopt}
	The \emph{strong opt-distance} between two left-hand side $\sigma$-structures $\A$ and $\B$ is defined as
	\begin{equation*}
		\dopt(\A, \B) \coloneqq \inf \set{\eps \mid \A \overcasts e^{-\eps} \B \text{ and } \B \overcasts e^{-\eps} \A}.
	\end{equation*}
\end{definition}
\begin{observation}
	Using the fact that $\maxval(\lambda\A,\C)=\lambda \maxval(\A,\C)$,
	it is an easy exercise to see that $\dopt(\A, \B) = \infty$
	if exactly one of $\maxval(\A,\C),\maxval(\B,\C)$ is $-\infty$,
	or exactly one of them is $0$, for some \MaxSol $\sigma$-structure $\C$;
	otherwise
	\begin{align*}
		\dopt(\A, \B) = \sup \set{\eps \mid \A \not\overcasts e^{-\eps} \B \text{ or } \B \not\overcasts e^{-\eps} \A} =
		\sup_{\C} \abs{\ln \maxval(\A,\C) - \ln \maxval(\B,\C)}.
	\end{align*}
	where the latter supremum is over all \MaxSol $\sigma$-structures $\C$
	such that neither is $-\infty$ nor $0$.
	It follows that $\dopt$ is symmetric and satisfies the triangle inequality.
\end{observation}

The only graph parameter $\p$ we consider in this paper will be treewidth, $\tw$.
Just as in~\cite{rwz21:soda}, one can prove that treedepth, or the Hadwiger number, give rise to equivalent definitions.

\begin{definition}
	For a graph parameter $\p$, a class of $\QQ_{\geq 0}$-valued structures $\cA$ is \emph{strongly $\p$-pliable} (with \emph{rate} $k(\eps)$) if for all $\eps > 0$ there exists $k=k(\eps)$ such that for every $\sigma$-structure $\A \in \cA$ there exists a $\QQ_{\geq 0}$-valued $\sigma$-structure $\B$ with $\p(\B) \leq k$ and $\dopt(\A, \B) \leq \eps$.
\end{definition}

\subsection{Duality}

\begin{definition}[partial functions and homomorphisms]\label{def:phom-monotone}
	For a partial function $g \colon A \to B$ and a tuple $\vx \in A^n$,
	$g(\vx)$ is defined as $(g(x_1),\dots,g(x_n)) \in B^n$ if all coordinates are defined, and is undefined otherwise.
	For $\vy \in B^n$, we define $g^{-1}(\vy) \coloneqq \{ \vx \in A^n \mid g(\vx) \text{ is defined and equal to } \vy\}$.
	
	For left-hand side $\sigma$-structures $\A,\B$, 
	a \emph{partial homomorphism} from $\A$ to $\B$ is a partial function $g \colon A \to B$ such that:
  for any positive tuple $(f,\vx) \in \tup_{>0}(\A)$, there is a positive tuple $(f,\vy) \in \tup_{>0}(\B)$ such that $y_i = g(x_i)$ whenever $g(x_i)$ is defined (and $y_i$ is arbitrary otherwise --- in particular $y_i \neq y_j$ is allowed even if $x_i=x_j$).
	We denote the set of partial homomorphisms from $\A$ to $\B$ by $\phom(\A,\B)$.
\end{definition}

\begin{remark}
	Partial homomorphisms can also be understood as follows.
	For a left-hand side $\sigma$-structure~$\B$,
	let $\B^+$ be the left-hand side $\sigma$-structure with domain $B \cup \{\star\}$, where~$\star$~is a new element,
	where the value for $f \in \sigma$ of arity $n$ and an input $\vx \in {(B \cup \{\star\})}^n$ is defined as
	\[f^{\B^+}(\vx) \coloneqq \max_{\substack{\vy \in B^n \\ \vx \poleq[\star] \vy}} f^{\B}(\vy).\]
	In particular $f^{\B^+}(\vx) = f^{\B}(\vx)$ for $\vx \in B^n$.
	Let $\Pos(\A)$ be the relational $\sigma$-structure consisting of positive tuples of $\A$. 
	Then a partial homomorphism $g$ from $\A$ to $\B$ is the same as a homomorphism from $\Pos(\A)$ to $\Pos(\B^+)$ (undefined assignments are the same as assignments to $\star$).
\end{remark}

\begin{lemma}\label{lem:duality}
	Let $\A, \B$ be left-hand side $\sigma$-structures.
	Then, the following are equivalent:
	\begin{itemize}
		\item $\A$ strongly overcasts $\B$, i.e.\ for all \MaxSol $\sigma$-structures $\C$, $\maxval(\A, \C) \geq \maxval(\B, \C)$;
		\item there is a distribution of partial homomorphisms $\omega \colon
      \phom(\A, \B) \to \Qpos$ $(\sum_g \omega(g) = 1)$ such that
		\[
			\Eg \fA(g^{-1}(\vy)) \geq \fB(\vy)
			\quad \quad \text{ for all } (f,\vy) \in \tup(\B).
		\]
	\end{itemize}
	(Here $\fA(g^{-1}(\vy))$ is a shorthand for $\sum\fA(\vx)$ over all $\vx \in g^{-1}(\vy)$, i.e.~all $\vx \in A^{\ar(f)}$ such that $g(\vx)$ is defined and equal to $\vy$.)
\end{lemma}

We shall call a distribution $\omega$ from the second bullet point a \emph{strong overcast}.

\begin{proof}
	For one direction, suppose there is a distribution $\omega$ as in the second bullet and let $\C$ be a \MaxSol $\sigma$-structure with a bottom label $c_{\bot}$.
	Let $h \colon B \to C$ be a function achieving $\maxval(\B, \C)$.
	For $g \in  \phom(\A, \B)$, let $h \circbot g \colon A \to C$ denote the function which maps $a \in A$ to $h(g(a))$ if $g(a)$ is defined and to $c_{\bot}$ otherwise.
	
	Therefore,
	\[
		\maxval(\A, \C)
		\geq \Eg \val(h \circbot g)
		= \Eg \sum_{(f,\vx) \in \tup(\A)} \fA(\vx) \fC(h \circbot g(\vx)).
		\label{eq:D_left}\tag{$\star_L$}
	\]

	\noindent
	We claim the expression \cref{eq:D_left} is greater or equal to
	\[
		\maxval(\B, \C)
		= \sum_{(f,\vy) \in \tup(\B)} \fB(\vy) \fC(h(\vy)).
		\label{eq:D_right}\tag{$\star_R$}
	\]
	Indeed, suppose first that \cref{eq:D_left} is $-\infty$, or equivalently, some summand in \cref{eq:D_left} is negative.
  Then there exists $(f,\vx) \in \tup_{>0}(\A)$ and $g \in \supp(\omega)$ with $f^\C(h\circbot g(\vx)) = -\infty$.
	Since $g$ is a partial homomorphism, by definition there exists a positive tuple $\vy \in B^{\ar(f)}$ such that $y_i = g(x_i)$ whenever $g(x_i)$ is defined.
	That is, $f^\B(\vy) > 0$ and $h \circbot g(\vx) \poleq h(\vy)$.
	By the assumption that $\C$ is a \MaxSol structure, $f^\C(h(\vy)) = -\infty$.
  Since $f^\B(\vy)$ is positive, this gives a $-\infty$ summand in \cref{eq:D_right} and thus the inequality holds.
	
	Otherwise, we can assume that every summand in \cref{eq:D_left} is non-negative.
	In that case
	\begin{align*}
		\cref{eq:D_left} &= \Eg \sum_{(f,\vx) \in \tup(\A)} \fA(\vx) \fC(h \circbot g(\vx))\\
		&\geq \Eg \sum_{\substack{(f,\vx) \in \tup(\A)\\g(\vx)\text{ is defined}}} \fA(\vx) \fC(h \circbot g(\vx))\\
		&= \Eg \sum_{(f,\vy) \in \tup(\B)} \fA(g^{-1}(\vy)) \fC(h(\vy))\\		
		&= \sum_{(f,\vy) \in \tup(\B)} \fC(h(\vy)) \Eg \fA(g^{-1}(\vy))\\
		& \geq \sum_{(f,\vy) \in \tup(\B)} \fC(h(\vy)) \fB(\vy) = \cref{eq:D_right},
	\end{align*}
	where after the first inequality it is still true that all summands are non-negative,
	and hence the last inequality follows from the assumption about $\omega$.
	This concludes the proof that $\maxval(\A,\C) \geq \maxval(\B,\C)$.	
	
	\bigskip	
	
	For the converse, we will use the following variant of Farkas' Lemma~\cite[Lemma A.2]{rwz20:arxiv-v3}:
	\emph{
		\mbox{Let $A$ be an $m \times n$ rational matrix and $\bar{b} \in \QQ^m$.}
		Exactly one of the following holds:
		\begin{itemize}
			\item there are $x_i\in\QQ_{\geq 0}$ ($i=1,\dots,n$) such that $\sum_i x_i = 1$ and $\sum_i A_{i,j} x_i \geq b_j$ for $j=1,\dots,m$;
			\item there are $y_j\in\QQ_{\geq 0}$ ($j=1,\dots,m$) such that $\sum_j A_{i,j} y_j < \sum_j b_j y_j$ for $i=1,\dots,n$.
		\end{itemize}
	}

	Suppose there is no distribution $\omega$ as in the second bullet.
	This means there are no numbers $\omega(g) \in \Qpos$ (for $g \in \phom(\A, \B)$)
	such that $\sum_g \omega(g) = 1$ and

	\[ \sum_{g \in \phom(\A, \B)} \omega(g) \fA(g^{-1}(\vy)) \geq \fB(\vy) \quad \quad \text{ for all } (f,\vy) \in \tup(\B).\]

	Thus by Farkas' Lemma,
	there are numbers $c(f,\vy) \in \Qpos$ (for $(f,\vy) \in \tup(\B)$) such that
	\begin{align}\label{eq:farkas}
		\sum_{(f,\vy) \in \tup(\B)} \fA(g^{-1}(\vy))\,c(f,\vy) < \sum_{(f,\vy) \in \tup(\B)} \fB(\vy) c(f,\vy)
		\quad \quad \text{for all }g \in \phom(\A, \B).
	\end{align}

	Let $\C$ be the $\sigma$-structure with domain $B \cup \{c_\bot\}$, where $c_\bot$ is a new label, and with values
	defined as follows for $f \in \sigma$ of arity $n$ and $\vy \in C^n$:
	\begin{align*}
		\fC(\vy) \coloneqq
		\begin{cases}
		-\infty &  \text{if } \forall_{\vy' \in B^n \colon  \vy' \pogeq \vy}\ \fB(\vy') = 0\\			
		& \text{ (in particular if } \vy \in B^n \text{ and } f^\B(\vy)=0 \text{)}\\
		c(f,\vy) & \text{if } \vy \in B^n \text{ and } \fB(\vy) > 0\\
		0 & \text{otherwise; that is, if\  } c_\bot \in \vy \text{ and } \exists_{\vy' \in B^n \colon \vy' \pogeq \vy}\ \fB(\vy') > 0.
		\end{cases}
	\end{align*}

	We claim that $\C$ is a \MaxSol structure. By Definition~\ref{def:maxsol}, we need to show that for each $f\in\sigma$, 
	whenever $f^\C(\vy)\geq 0$, we have $f^\C(\vx)\geq$ for all $\vx \poleq \vy$.
	Let $f^\C(\vy)=c(f,\vy)$ (second case in the definition of $f^\C$) and $\vx\poleq\vy$. If $\vx\in B^n$ then $f^\C(f,\vx)=c(f,\vx)$ and the claim holds as $c(f,\vx)\geq 0$. If $c_\bot\in\vy$ then $f^\C(\vy)=0$ and the claim holds again. Finally, if $f^\C(\vy)=0$ from the third case in the definition of $f^\C$, then for any $\vx\poleq\vy$
	we have $f^\C(\vx)=0$ (from the third case).
	To show that $\maxval(\A, \C) < \maxval(\B, \C)$,
	we claim that for every function $g$ from $A$ to $C = B \cup \{c_\bot\}$, we have the following strict inequality:
\begin{align*}
		\val(g) &= \sum_{(f,\vx) \in \tup(\A)} \fA(\vx) \fC(g(\vx))
		\ =\ \sum_{(f,\vy) \in \tup(\C)} \fA(g^{-1}(\vy)) \fC(\vy) \\
		&< \sum_{(f,\vy) \in \tup(\B)} \fB(\vy) \fC(\vy)
		\ =\ \val(\iota)\ \leq\ \maxval(\B, \C),
	\end{align*}
	where $\iota$ denotes the inclusion function from $B$ to $C$.
	
	Indeed, suppose first that $g$, as a partial function from $A$ to $B$, is not a partial homomorphism.
  Then there is an $(f,\vx) \in \tup_{>0}(\A)$ such that for all $\vy \in B^{\ar(f)}$ with $\vy \pogeq g(\vx)$ we have  $\fB(\vy)=0$.
	Thus $\fC(g(\vx))=-\infty$ by definition.
	Thus the left-hand side of the inequality is $-\infty$, while the right-hand side is always non-negative.
	
	Otherwise, we have
	\begin{align*}
		\sum_{(f,\vy) \in \tup(\C)} \fA(g^{-1}(\vy)) \fC(\vy)
		\leq \sum_{(f,\vy) \in \tup(\B)} \fA(g^{-1}(\vy)) \fC(\vy)
		\leq \sum_{(f,\vy) \in \tup(\B)} \fA(g^{-1}(\vy)) c(f,\vy)\\
		< \sum_{(f,\vy) \in \tup(\B)} \fB(\vy) c(f,\vy) = \sum_{(f,\vy) \in \tup(\B)} \fB(\vy) \fC(\vy).
	\end{align*}
	The first inequality follows from the fact that for $(f,\vy) \in \tup(\C) \setminus \tup(\B)$ we have $\fC(\vy) \leq 0$.
	The second follows from $f^\C(\vy) \leq c(f,\vy)$.
	The third, strict inequality follows from~\cref{eq:farkas} since $g$ (as a partial function from $A$ to $B$) is a partial homomorphism.
	The final equality follows from the fact that if $f^\C(\vy) \neq c(f,\vy)$ for $(f,\vy) \in \tup(\B)$, then $f^\B(\vy)=0$.
	This concludes the proof that $\maxval(\A, \C) < \maxval(\B, \C)$.
\end{proof}

\subsection{PTAS}

We first define the Sherali-Adams LP hierarchy~\cite{sa90} for \MaxSol.
Let $(\A,\C)$ be an instance of \MaxSol over a signature $\sigma$ and let $k\geq\max_{f\in\sigma}\ar(f)$.
We write $\binom{A}{\leq k}$ for the set of subsets of $A$ with at most $k$ elements.
The \emph{Sherali-Adams relaxation of level $k$}~\cite{sa90} of $(\A,\C)$ is the
linear program given in \cref{fig:sa}, denoted by $\SA{k}(\A,\C)$, which has
one variable $\lambda(X,s)$ for each $X\in\binom{A}{\leq k}$ and each $s\colon X \to C$.
We denote by $\maxval_k(\A,\B)$ the optimum value of $\SA{k}(\A,\C)$, and define
$\maxval_k(\A,\B)=-\infty$ if $\SA{k}(\A,\C)$ is infeasible.

\begin{figure*}[t]
	\centering
	\fbox{\parbox{.98\textwidth}{
			\begin{align*}
			\max \sum_{(f,\vx) \in \tup{\A}, \; s\colon \toset{\vx} \to C}& \lambda(\toset{\vx},s) \fA(\vx) \fC(s(\vx)) \\
			\lambda(X,s) &= \sum_{r\colon Y\to C,\, r|_{X}=s} \lambda(Y,r) & \text{for $X \subseteq Y\in\textstyle\binom{A}{\leq k}$ and $s\colon X\to C$}  \\
			\sum_{s\colon X \to C} \lambda(X,s) &= 1   &\text{for $X\in\textstyle\binom{A}{\leq k}$} \\
			\lambda(\toset{\vx},s) &= 0   &\text{$\forall (f,\vx)\in\tup(\A)$ with $\fA(\vx)\fC(s(\vx))=-\infty$} \\
			\lambda(X,s) &\geq 0   &\text{for $X\in\textstyle\binom{A}{\leq k}$ and $s\colon X \to C$}
			\end{align*}
	}}
	\caption{$\SA{k}(\A,\C)$, the Sherali-Adams relaxation of level $k \geq \max_{f\in\sigma} \ar(f)$ of $(\A,\C)$.}\label{fig:sa}
\end{figure*}

\begin{observation}\label{obs:linearity}
	Let $(\A,\C)$ be an instance of \MaxSol, $k \geq \max_{f\in\sigma} \ar(f)$ and $\lambda \geq 0$.
	Then,
	$\maxval(\lambda \A,\C) = \lambda \maxval(\A, \C)$
	and
	$\maxval_k(\lambda \A,\C) = \lambda \maxval_k(\A, \C)$.
\end{observation}

\begin{observation}\label{obs:SA-relax}
	Let $(\A,\C)$ be an instance of \MaxSol. Then, for any $k \geq \max_{f \in \sigma} \ar(f)$,
	$\maxval_k(\A, \C) \geq \maxval(\A, \C)$.
\end{observation}
\begin{proof}
	Let $h \colon A \to C$ be an optimal solution to $(\A,\C)$.
	Consider the solution $\lambda(X,s)=\I{s=\restr h X}$ for $\SA{k}(\A,\C)$.
	It is trivially feasible and achieves the value $\maxval(\A,\C)$.
\end{proof}

The following easy result (proved in \cref{sec:SA-tw}) shows that an
appropriate level of the Sherali-Adams relaxation is exact for bounded treewidth.

\begin{proposition}\label{prop:SA-tw}
	Let $(\A,\C)$ be an instance of \MaxSol and $k \geq \tw(\A)$.
	Then,
	$\maxval_k(\A, \C) = \maxval(\A, \C)$.
\end{proposition}

\begin{definition}\label{def:overcasts_k}
	Let $\A$ and $\B$ be left-hand side $\sigma$-structures, and $k \geq \max_{f\in\sigma} \ar(f)$.
	We write $\A \overcasts_k \B$ if for all \MaxSol $\sigma$-structures $\C$ we have
	$\maxval_k(\A,\C) \geq \maxval_k(\B,\C)$.
\end{definition}

Using the dual characterisation of strong overcasts (\cref{lem:duality}), we can show
(and prove in \cref{sec:overcast-sa}) the following. 

\begin{proposition}\label{prop:overcast-sa}
	Let $\A$ and $\B$ be left-hand side $\sigma$-structures, and $k \geq \max_{f\in\sigma} \ar(f)$.
	If $\A \overcasts \B$, then $\A \overcasts_k \B$.
\end{proposition}

We are now ready to prove our main tractability result for maximisation
problems.

\begin{lemma}\label{lem:tract}
	Let $\A$ be a left-hand side $\sigma$-structure,
	$\eps \geq 0$ be small and $k\geq\max_{f\in\sigma}\ar(f)$.
	Suppose that there exists a left-hand side $\sigma$-structure $\B$
	such that $\dopt(\A, \B) \leq \eps$ and $\tw(\B)\leq k$.
	Then, for every right-hand side $\sigma$-structure $\C$, we have that
	\[ \maxval(\A,\C) \leq \maxval_k(\A,\C) \leq (1+\BigO{\eps}) \maxval(\A, \C). \]
\end{lemma}
\begin{proof}
	By definition of $\dopt$ we have that,\vspace*{-.3\baselineskip}
	\[ \A \overcastsL e^\eps \B \overcastsL e^{2\eps} \A, \vspace*{-.5\baselineskip}\]
	and so\vspace*{-.4\baselineskip}
	\[ \A \overcastsL_k e^\eps \B \overcastsL_k e^{2\eps} \A \]
	by \cref{prop:overcast-sa}.
	From \cref{obs:SA-relax,prop:SA-tw,obs:linearity} we obtain that,
  \[ \maxval(\A, \C) \leq \maxval_k(\A, \C) \leq e^{\eps} \maxval_k(\B, \C) =
  e^{\eps} \maxval(\B,\C) \leq e^{2\eps} \maxval(\A, \C).\]
	Finally, for $\eps$ small we have $e^{2\eps} = 1+\BigO{\eps}$, completing the proof.
\end{proof}
 
Since $\maxval_k(\A,\C)$ can be computed in time ${(\card{\A} \cdot
\card{\C})}^{\BigO{k}}$, we obtain that any strongly $\tw$-pliable class of
structures admits a PTAS.

\begin{corollary}\label{cor:tract}
  Let $\cA$ be a strongly $\tw$-pliable class of left-hand side structures. Then, the class of
  Max-Sol instances $(\A,\C)$ with $\A\in\cA$ admits a PTAS.
  
  Specifically, if $\cA$ is strongly $\tw$-pliable with rate $k(\eps)$,
  then given $\A\in\cA,\C,$ and $\eps>0$,
  we can output a value 
  between $\maxval(\A,\C)$ and $(1+\eps) \maxval(\A, \C)$
  in time ${(\card{\A} \cdot \card{\C})}^{\BigO{k(\Omega(\eps))}}$.
\end{corollary}

In the following subsection, we show that when we look at the class of Gaifman graphs only, the appropriate condition is fractional-treewidth-fragility.

\subsection{Fragility and pliability}

To give Dvo\v{r}\'ak's definition of fractional fragility~\cite{Dvorak16}
we first define $\eps$-thin distributions.

\begin{definition}
	Let $\mathcal{F}$ be a family of subsets of a set $V$ and $\eps > 0$.
	We say a distribution $\pi$ over $\mathcal{F}$ is \emph{$\eps$-thin} if $\Pr_{X \sim \pi} [v \in X] \leq \eps$ for all $v \in V$.
\end{definition}

We now give some intuition for the next definition. Consider the treewidth as a graph parameter. The idea of a \emph{modulator}, defined below, is to remove a subset $X$ of the vertices of a graph $G$ to obtain a bound on the treewidth of the new graph $G-X$. The fractional variant considers a distribution over modulators. 
An alternative view of fractional-$\tw$-fragility (obtained by LP duality~\cite{DvorakS20}) is that for any $\eps>0$ there is $k$ such that for any vertex weight function on $G$, a removal of a set vertices $X$ whose weight is an $\eps$-fraction of the total weight yields a graph $G-X$ of treewidth at most $k$.

\begin{definition}\label{def:fragile}
  For a graph parameter $\p$ and a number $k$, we define a
  \emph{$(\p\leq k)$-modulator} of a graph $G$ to be a set $X \subseteq
  V(G)$ such that $\p(G-X) \leq k$. A \emph{fractional $(\p\leq
  k)$-modulator} is a distribution $\pi$ of such modulators $X$.
  We say that a class of
  graphs $\cG$ is \emph{fractionally-$\p$-fragile} (with \emph{rate} $k(\eps)$)
  if for every $\eps>0$ there is a $k=k(\eps)$ such that every $G \in \cG$ has an $\eps$-thin fractional
  $(\p\leq k)$-modulator. 
\end{definition}

We need some more notation. We denote the disjoint union of graphs $G$ and $H$
by $G \uplus H$.
For $\sigma$ structures $\A_1,\dots,\A_k$, we define the $\sigma$-structure $\B =
\biguplus_{i=1}^k \A_i$ to be over the domain $B=\biguplus_{i=1}^k A_i$ and by
$\fB(\vx) = f^{\A_i}(\vx)$ whenever $(f,\vx) \in \tup(\A_i)$, and 0 otherwise.

While we are mostly interested in the following result with treewidth as the
graph parameter, we state it more generally since the proof is the same.

\begin{lemma}\label{lem:fragility-pliable}
	Let $\p$ be a monotone\footnote{$\p(H) \leq \p(G)$ for all graphs $G$ and
  subgraphs $H$ of $G$.} graph parameter such that $\p(G \uplus H) \leq
  \max(\p(G), \p(H))$ for all graphs $G$ and $H$ and $\p(G) \leq
  \p(G-v) + 1$ for all $v \in V(G)$.
	Let $\cA$ be a class of structures with bounded arity $r$ such that the class
  $\cG$ of their Gaifman graphs is fractionally-$\p$-fragile (with rate $k(\eps)$).
	Then $\cA$ is strongly $\p$-pliable (with rate $k'(\eps) = k(\Omega(\eps/r))+1$). 
\end{lemma}

The proof closely follows the proof of~\cite[Lemma 4.6]{rwz20:arxiv-v3}, where
the same result was shown for several particular monotone graph parameters.

\begin{proof}
	Given $\eps > 0$, $\A \in \cA$, let $\pi$ be a fractional $(\p \leq k)$-modulator such that for every $v \in V(G)$,
	\begin{equation}\label{eq:probX}
	\Pr_{X\sim\pi}[v \in X] \leq \eps.
	\end{equation}
	For each $X \subseteq V(G)=A$ in the support of $\pi$ ($\pi(X) > 0$), define $\A_{/X}$ to be the $\sigma$-structure obtained by contracting $X$ to a single vertex and summing values.
	That is, let $\{\star_X\}$ be a new element and define $g_X: A \to (A-X) \cup \set{\star_X}$ that maps $X$ to $\star_X$ and $A-X$ identically.
	Let $\A_{/X}$ be over the domain $(A-X) \cup \{\star_X\}$ and
		\[f^{\A_{/X}}(\vy) \coloneqq \fA(g_X^{-1}(\vy)) = \sum_{\vx \in g_X^{-1}(\vy)} f^\A(\vx)\]
	for each $f \in \sigma$ of arity $n$ and each $\vy \in {((A-X)\cup\{\star_X\})}^n$.
	
	Define $\B_X = \pi(X) \cdot \A_{/X}$, and let $\B = \biguplus \B_X$.
	By definition of $\pi$ and properties of $\p$,
	we have $\p(\Gaifman(\B_X)) \leq \p(\Gaifman(\A)-X) + 1 \leq k+1$, and so
  $\p(\Gaifman(\B)) \leq k+1$.

	View $g_X$ as a function to $B$ (instead of as function to $B_X \subseteq B$), so that $g_X \colon A \to B$ is the (total) function mapping $A-X$ identically to its copy in $B_X$ and mapping~$X$~to~$\star_X$.
	It is clear that $g_X \in \phom(\A, \B)$.
	Define the strong overcast $\omega \colon \A \to \B$ to take the value $g_X$ with probability $\pi(X)$.
	To check this is indeed a strong overcast, observe that for $(f,\vy) \in \tup_{>0}(\B)$,
	there is a unique $X$ such that $(f,\vy) \in \tup(\B_X)$,
	hence
	\[\Eg \fA(g^{-1}(\vy)) = \pi(X) \fA(g_X^{-1}(\vy)) = \fB(\vy).\]

	Define $g \colon B \to A$ to be the partial function mapping each element of $B_X - \{\star_X\}$ identically to $A$, leaving it undefined on $\star_X$.
	It is clear that $g \in \phom(\B, \A)$.
	Consider the overcast $\omega' \colon \B \to (1-r\eps)\A$ that is deterministically $g$.
	To check that $\omega'$ is indeed a strong overcast, let $(f,\vx) \in \tup(\A)$.
	Then $\vx$ is covered by copies in $\B_X$ for those $X$ that do not intersect $\vx$, hence
	\begin{align*}
	\fB(g^{-1}(\vx))
	&= \EX_{X \sim \pi} \brk{\I{X \cap \vx = \emptyset} \cdot \fA (\vx)} \\
	&= \fA(\vx) \Pr_{X \sim \pi} (X \cap \vx = \emptyset) \\
	&\geq \fA(\vx) \cdot (1-r\eps),
	\end{align*}
	where the final inequality follows by~\cref{eq:probX}, the union bound, and the fact that $\card{\vx} \leq r$.
	Hence by \cref{lem:duality} applied to $\omega$ and $\omega'$,
	\[ \A \overcasts \B \overcasts (1-r\eps) \A. \]
	By construction $\p(\Gaifman(\B)) \leq k+1$.
	Thus we have shown that assuming Gaifman graphs of structures in $\cA$ are fractionally-$\p$-fragile with rate $k(\eps)$,
	then for every $\eps$ and for every structure $\A \in \cA$
	there is a structure $\B$ with $\p(\Gaifman(\B)) \leq k(\eps)+1$ and $\dopt(\A, \B) \leq \BigO{r\eps}$.
	As $r$ is fixed, this implies that $\cA$ is strongly $\p$-pliable (with rate $k'(\eps) = k(\Omega(\eps/r))+1$).
\end{proof}

\begin{remark}
In the above lemma, the assumption that $\cA$ contains structures of bounded arity $r$ can be easily lifted, at least for $\p=\tw$.
This is because the maximum arity of a structure in $\cA$ is bounded by the size of the largest clique in a Gaifman graph of a structure in $\cA$.
Since we assume that the class of their Gaifman graphs is fractionally-$\tw$-fragile with rate $k(\eps)$, the largest clique has at most $2k(\frac{1}{2})+2$ vertices (otherwise any $(\tw\leq k(\frac{1}{2}))$-modulator needs to contain more than half of the clique's vertices, and there cannot be a $\frac{1}{2}$-thin  distribution of such modulators).
Thus without loss of generality we can assume $r \leq 2k(\frac{1}{2})+2$.
\end{remark}

\begin{proof}[Proof of \cref{thm:main-maxsol}]
  Let $\cG$ be a class of graphs that is fractionally-treewidth-fragile 
  and let $\cA$ be a class of structures with bounded arity with Gaifman graphs
  in $\cG$. Since treewidth satisfies the assumptions of \cref{lem:fragility-pliable}, 
  we have that $\cA$ is strongly $\tw$-pliable. By \cref{cor:tract}, 
  $\MaxSol_\cG$ admits a PTAS.
\end{proof}

If we only look at Gaifman graphs, one cannot use the presented approach to go beyond fractionally-treewidth-fragile classes. This is because~\cite[Lemma 6.1]{rwz20:arxiv-v3} together with the above \cref{lem:fragility-pliable} implies that for a class of graph $\cG$ and an integer $r$, if $\cA^{(r)}_\cG$ denotes the class of all $\QQ_{\geq 0}$-valued structures of arity at most $r$ and whose Gaifman graphs are in $\cG$, then $\cA^{(r)}_\cG$ is strongly $\tw$-pliable if and only if  $\cG$ fractionally-treewidth-fragile.
In \cref{sec:dense}, 
we give a simple example of a class of structures (not
parametrised by their Gaifman graphs) that
is strongly tw-pliable but not captured by fractional-treewidth-fragility.

{
\small
\bibliographystyle{plainurl}
\bibliography{mwz}
}

\appendix

\section{Proof of~\cref{prop:overcast-sa}}\label{sec:overcast-sa}
We closely follow the proof of~\cite[Proposition~5.3]{rwz20:arxiv-v3} but, given
we are in a more general setting, we have to be more careful.

\begin{proof}
	Let $\C$ be a \MaxSol $\sigma$-structure, and let $\omega$ be an overcast from $\A$ to $\B$.
	Recall that for a tuple $\vx$ we denote by $\toset{\vx}$ the set of elements appearing in $\vx$.
	For a partial function $g \colon A \to B$ and a subset $X \subseteq A$,
	$g(X)$ denotes the set $\{g(x) \mid x \in X\text{ and }g(x)\text{ is defined}\}$.
	For a function $s \colon g(X) \to C$, $s\circbot g$ denotes the function from $X$ to $C$ defined as $s(g(x))$ when $g(x)$ is defined and $c_\bot$ otherwise.

	Let $\lambda$ be an optimal solution to $\SA{k}(\B,\C)$.
	That is, for each subset $Y$ of $B$ of size at most~$k$, $\lambda$ describes a distribution of functions to $C$
	using probabilities $\lambda(Y, s) \in \QQ_{\geq 0}$ for $s \colon Y \to C$.
	For fixed $g \in \supp(\omega)$, we define a solution $\lambda_g$ to $\SA{k}(\A,\C)$ by sampling $s$ from this distribution and outputting $s \circbot g$.
	Formally, for $X \in \binom{A}{\leq k}$ and $r \colon X \to C$, we define
	\[ \lambda_g(X,r) \coloneqq \sum_{s \colon g(X) \to C} \I{r = s \circbot g} \cdot \lambda(g(X), s). \]
	Note that $\lambda_g$ is a feasible solution.
	Indeed, for $X \in \binom{A}{\leq k}$ the total probability is
	\[ \sum_{r \colon X \to C} \lambda_g(X,r) = \sum_{s \colon g(X) \to C} \left(\sum_{r \colon X \to C} \I{r = s \circbot g}\right) \cdot  \lambda(g(X), s) = \sum_{s \colon g(X) \to C} 1 \cdot \lambda(g(X), s) = 1;\]
	while for $Z\subseteq X\in \binom{A}{\leq k}$ and $r\colon Z \to C$,
	the marginal probability of obtaining $r$ is
	\begin{alignat*}{3}
	\sum_{\substack{t\colon X \to C\\t|_Z = r}}\lambda_g(X,t)
	&= \sum_{\substack{t\colon X \to C\\t|_Z = r}}\ \sum_{s \colon g(X) \to C} \I{t=s\circbot g} \cdot \lambda(g(X),s)   \\
	&= \sum_{s \colon g(X) \to C}\ \sum_{\substack{t\colon X \to C\\t|_Z = r}} \I{t=s\circbot g} \cdot \lambda(g(X),s)   \\
	&= \sum_{s \colon g(X) \to C} \I{r = (s\circbot g)|_Z} \cdot \lambda(g(X),s)   \\
	&=  \sum_{s \colon g(X) \to C} \I{r = (s|_{g(Z)}) \circbot g} \cdot \lambda(g(X),s)\\
	&= \sum_{s' \colon g(Z) \to C} \Big(\I{r = s' \circbot g} \sum_{\substack{s \colon g(X) \to C\\s|_{g(Z)} = s'}} \lambda(g(X),s)\Big)\\
	&= \sum_{s' \colon g(Z) \to C} \I{r = s' \circbot g} \cdot \lambda(g(Z),s')
	&= \lambda_g(Z,r).
	\end{alignat*}
	Therefore $\maxval_k(\A,\C)$ is at least the expected value of the solution $\lambda_g$ with $g$ sampled from $\omega$:
	\begin{align}
	\maxval_k(\A,\C) &\geq \EX_{g \sim \omega} \sum_{\substack{(f,\vx) \in \tup(\A)\\r\colon\toset{\vx} \to C}} \lambda_g(\toset{\vx},r) \fA(\vx) \fC(r(\vx)) \nonumber\\
	&=\EX_{g \sim \omega}  \sum_{\substack{(f,\vx) \in \tup(\A)\\s\colon g(\toset{\vx}) \to C}} \lambda(g(\toset{\vx}),s) \fA(\vx) \fC(s \circbot g(\vx)),  \label{eq:SA_left}\tag{$\ast_L$}
	\end{align}
	by definition of $\lambda_g$.
	We claim the expression \cref{eq:SA_left} is at least
	\[\maxval_k(\B,\C) = \sum_{\substack{(f,\vy) \in \tup(\B)\\s\colon\toset{\vy} \to C}} \lambda(\toset{\vy},s) \fB(\vy) \fC(s(\vy)). \label{eq:SA_right}\tag{$\ast_R$}\]

	Indeed, suppose first that some summand in \cref{eq:SA_left} is negative.
	Then there exist $g \in \supp(\omega)$, $(f,\vx) \in \tup(\A)$, and $s \colon g(\toset{\vx}) \to C$ such that
	$\lambda(g(\toset{\vx}), s), \fA(\vx) > 0$ and $\fC(s \circbot g(\vx)) = -\infty$.
	Since $g \in \phom(\A,\B)$, there is some $\vy \in A^{\ar(f)}$ with $\fB(\vy) > 0$ such that $g(x_i)$ equals $y_i$ whenever it is defined.
	In particular $g(\toset{\vx}) \subseteq \toset{\vy}$, hence $\lambda(\toset{\vy},s') > 0$ for some $s' \colon \toset{\vy} \to C$ such that $s'|_{g(\toset{\vx})} = s$.
	This implies $s'(\vy)  \pogeq s\circbot g(\vx)$, hence $\fC(s'(\vy)) = -\infty$ by the assumption that $\C$ is a $\MaxSol$ structure.
	Therefore \cref{eq:SA_right} has a summand $\lambda(\toset{\vy},s')\fB(\vy)\fC(s'(\vy)) = -\infty$,
	so the claimed inequality \cref{eq:SA_left} ${}\geq{}$ \cref{eq:SA_right} holds.

	Otherwise, we can assume that every summand in \cref{eq:SA_left} is non-negative.
	In that case

	\begin{align*}
	\cref{eq:SA_left} = & \Eg \sum_{(f,\vx) \in \tup(\A), \, s\colon g(\toset{\vx}) \to C} \lambda(g(\toset{\vx}),s) \fA(\vx) \fC(s \circbot g(\vx))  \\
    \geq& \Eg \sum_{(f,\vx) \in \tup(\A), \, s\colon g(\toset{\vx}) \to C} \lambda(g(\toset{\vx}),s) \fA(\vx) \fC(s \circbot g(\vx)) \cdot \I{g(\vx)\text{ is defined}} \\
	=& \Eg \sum_{(f,\vy) \in \tup(\B), \, s\colon\toset{\vy} \to C} \lambda(\toset{\vy},s) \fA(g^{-1}(\vy)) \fC(s(\vy))\\
	=& \sum_{(f,\vy) \in \tup(\B),\, s\colon\toset{\vy} \to C}
	\lambda(\toset{\vy},s) \Eg \brk{\fA(g^{-1}(\vy))} \fC(s(\vy))\\
	\geq& \sum_{(f,\vy) \in \tup(\B), \, s\colon\toset{\vy} \to C} \lambda(\toset{\vy},s) \fB(\vy) \fC(s(\vy)) = \cref{eq:SA_right},
	\end{align*}
	where after the first inequality it is still true that all summands are non-negative,
	and hence the last inequality follows from the fact that $\omega$ is an overcast.
	This concludes the proof that $\maxval_k(\A,\C) \geq \maxval_k(\B,\C)$.
\end{proof}

\section{Proof of~\cref{prop:SA-tw}}\label{sec:SA-tw}
\begin{proof}
	We reduce to~\cite[Theorem~5.4]{crz22:sicomp}, which shows that bounded
  treewidth implies exact solvability for minimisation of VCSPs
  with $\Qposinf$-valued right-hand side structures.
  We recast our relaxation into the framework of~\cite{crz22:sicomp}, which gives
  a more fine-grained relaxation. 
  The SA relaxation in that paper is found in Figure 2. 
	As we restrict to $k \geq \max_{f\in\sigma}\ar(f)$, in this case variables
  $\lambda(f,\vx,s)$ in~\cite[Figure~3]{crz22:sicomp} may be replaced by variables $\lambda(\toset{\vx},s)$ by equation (SA3) and the inclusion of the dummy function $\rho_k$ in $\sigma_k$.
	The linear programs are now equivalent, except in~\cite{crz22:sicomp} right-hand
  side $\sigma$-structures are \Qposinf-valued, and it is cast as a minimisation problem.

	Let $K \coloneqq \max_{(f,\vx) \in \tup(\C)} \fC(\vx) \in \Qpos$.
	Consider a new $\sigma$-structure $\C'$ defined by
	\[ f^{\C'}(\vx) = K - \fC(\vx), \]
	so that $\C'$ is \Qposinf-valued, and thus a valued $\sigma$-structure in the framework of~\cite{crz22:sicomp}.

	We have,
	\begin{align*}
		\max \sum_{(f,\vx) \in \tup{\A}, \; s\colon \toset{\vx} \to C} \lambda(\toset{\vx},s) \fA(\vx) \fC(s(\vx)) = \\
		K \sum_{(f,\vx) \in \tup{\A}} \fA(\vx)
		-\min \sum_{(f,\vx) \in \tup{\A}, \; s\colon \toset{\vx} \to C} \lambda(\toset{\vx},s) \fA(\vx) f^{\C'}(s(\vx)).
	\end{align*}

	The left term does not depend on the LP variables $\lambda(X,s)$.
	Thus $\SA{k}$ relaxation is exact by~\cite[Theorem~5.4]{crz22:sicomp}.
\end{proof}

\section{Max-CSP vs. \MaxSol on cliques}\label{sec:dense}
In~\cite{rwz21:soda}, the PTAS results for Max-CSPs, with non-negative rational-valued
right-hand side structures, apply to many classes of dense structures as well.
This is because the class of cliques (as a $\{0,1\}$-valued structures) and in fact any class of graphs with $\Omega(n^2)$ edges, was shown to be ``$\tw$-pliable''.
In our case, because we consider more general right-hand side structures,
the definition of ``overcasts'' and ``$\tw$-pliability'' changed to ``strong overcasts'' and ``strong $\tw$-pliability'' accordingly (simply by considering all $\MaxSol$ right-hand side structures in place of just $\QQ_{\geq 0}$-valued structures).
It turns out even the simplest class of dense structures, the class of cliques, is not strongly $\tw$-pliable.

\begin{proposition}\label{prop:clique-not-tw-pliable}
	Let $\sigma$ be the signature with one symbol $f$ of arity 2.
	Let $\mathcal{A} = \{\mathbb{K}_n \mid n \in \NN \}$ be the class of cliques, as $\sigma$-structures with $f^{\mathbb{K}_n}(x,y) = 1$ if $x\neq y$ and $0$ otherwise, for $x,y \in [n]$.
	Then, $\cA$ is not strongly tw-pliable.
\end{proposition}
\begin{proof}
	Suppose that $\cA$ is strongly tw-pliable.
	Let $\eps=1/10$. There exists some $k$ such that for all $n$ there exists a structure $\B$ with $\tw(\B) \leq k-1$ and
	$\dopt(\bbK_n, \B) \leq \eps$.
	
	Let $n \coloneqq 2k$ and $\B$ be such that $\tw(\B) \leq k$ and $\dopt(\bbK_n, \B) \leq \eps$.
	It is an easy exercise that for any graph $G$, $\chi(G) \leq \tw(G)+1$, and therefore
  $\chi(\Gaifman(\B)) \leq k$, where $\chi(G)$ denotes the chromatic number of
  $G$.
	
	Consider the following class of \MaxSol structures: for each $i$, let $\C_i$ be a $\sigma$-structure over the domain $[i] \cup \set{c_\bot}$, and let
	\begin{equation*}
		f^{\C_i}(x,y) = \begin{cases}
			-\infty	&\qquad \text{if $x=y$, $x,y \in [i]$,} \\
			0	&\qquad \text{if $x=y=c_\bot$,} \\
			1		&\qquad \text{otherwise.}
		\end{cases}
	\end{equation*}
	In other words, for a structure $\A$ the instance $(\A,\C_i)$ asks to colour the vertices of $\A$ with $i$ colours (or assign it no colour, $c_\bot$), such that there are no monochromatic edges and the total weight of edges with at least one endpoint coloured is maximised.
	Then, $\maxval(\bbK_n,\C_i)=i(n-1)$ for $i \leq n$.
	Further, for all $i \geq \chi(\Gaifman(\B))$, $\maxval(\B,\C_i)=\maxval(\B,\C_{\chi(\Gaifman(\B))})$: the optimal solution corresponds to any proper colouring of $\Gaifman(\B)$ with colours $[i]$.
	As $k \geq \chi(\Gaifman(\B))$, we have that
	$\maxval(\bbK_n,\C_{2k})/\maxval(\bbK_n,\C_{k})=2$, but
	$\maxval(\B,\C_{2k})/\maxval(\B_n,\C_{k})=1$,
	contradicting $\dopt(\bbK_n,\B) \leq \eps$.
\end{proof}

In view of \cref{prop:clique-not-tw-pliable}, one may ask whether it
is possible to obtain a PTAS on cliques via different means, not relying on our notion of
strong tw-pliability. It turns ous that this is not possible.
This follows from an easy reduction from the Maximum Clique problem.

\begin{lemma}[{\cite[Theorem~1.1]{Zuckerman07:toc}}]\label{lem:max-clique-hard}
It is NP-hard to approximate Maximum Clique within a factor $\mathrm{opt}^{1-\eps}$ for any $\eps > 0$.
That is, unless P=NP, for any $\eps>0$ there is no polynomial-time algorithm taking a graph $G$ and an integer $r$ as input, that can distinguish between the following cases:
\begin{itemize}
	\item $G$ has a clique of size at least $r$,
	\item $G$ has no clique of size $r^\eps$.
\end{itemize}
\end{lemma}
\begin{proposition}\label{prop:max-dense}
	Let $\cA=\set{\bbK_n \mid n \in \mathbb N}$ be the class of cliques (as defined in \cref{prop:clique-not-tw-pliable}). 
	Then VCSP restricted to instances $(\A,\C)$ where $\A \in \mathcal{A}$ and $\C$ is a \MaxSol structure, does not admit a PTAS unless P=NP\@.
\end{proposition}
\begin{proof}
	We reduce an instance $G$ of Maximum Clique to a suitable VCSP instance.
	Given a graph $G$ on $n$ vertices, define a $\sigma$-structure $\C$ over the domain $V(G) \cup \set{\star}\ \cup \set{c_\bot}$ as follows. Let
	\begin{align*}
		\fC(x,y) = \begin{cases}
			-\infty	&\qquad \text{if $x,y \in V(G)$ and $xy \not\in E(G)$, or $x=y=\star$} \\
			1		&\qquad \text{if $x=\star$ and $y \in V(G)$} \\
			0		&\qquad \text{otherwise.}
		\end{cases}
	\end{align*}
	It is clear that $\C$ is a \MaxSol structure, and it is easy to see that $\maxval(\bbK_{n+1},\C)=\MaxClique(G)$: in any feasible homomorphism from $\bbK_{n+1}$ to $\C$, at most one vertex may map to each element in $V(G) \cup \set{\star}$, and vertices that are not mapped to $c_\bot$ or $\star$ have to map to a clique in $G$. The optimal solution is achieved by mapping exactly one vertex to $\star$, one to each vertex of a maximum clique in $G$ and the remaining to $c_\bot$.

	Thus by \cref{lem:max-clique-hard}, unless P=NP there is no polynomial-time algorithm approximating $\maxval(\bbK_{n+1},\C)$ within any sublinear factor,
	let alone a constant factor approximation or a PTAS\@.
\end{proof}

Finally, we give a simple example showing that there are strongly tw-pliable (and
therefore, by \cref{lem:tract}, tractable) classes of structures that are not captured by fractional-treewidth-fragility.
Note that this does not contradict our conjecture made after
\cref{thm:main-maxsol} in \cref{sec:intro}, as the conjecture is
restricted only to classes parametrised by their Gaifman graphs.  
The class of structures in the following proposition does \emph{not} include all possible structures over the Gaifman graphs.

\begin{proposition}
	Let $\sigma$ be the signature with one function symbol $f$ of arity 2.
	Let $\mathbb A_n$ be $\sigma$-structure over the domain $[n]$ defined by $f^{\A_n}(x,x)=1$ and $f^{\A_n}(x,y)=1/n$ if $x \neq y$, $x,y \in [n]$. (That is, $\A_n$ is clique with loops around each vertex, loops have weight 1, and simple edges have weight $1/n$.)
	Then, $\cA=\set{\A_n \mid n \in \mathbb N}$ is strongly tw-pliable but $\Gaifman(\cA)$ is not fractionally-treewidth-fragile.
\end{proposition}
\begin{proof}
	First note that $\Gaifman(\cA)$ is the class of cliques with a loop on each vertex.
	As for any (non-empty) $X \subseteq V(K_n)$ we have $\tw(K_n - X) = n-1-\card{X}$, it is easy to see that $\Gaifman(\cA)$ is not fractionally-treewidth-fragile --- alternatively, it follows from \cref{prop:clique-not-tw-pliable,lem:fragility-pliable} that the class of cliques is not fractionally-treewidth-fragile, and thus neither is the class of cliques with a loop around each vertex, i.e.\ $\Gaifman(\cA)$.
	
	We now show that $\cA$ is strongly tw-pliable. Let $\eps>0$ be small, $k \coloneqq \lceil 2/\eps \rceil$ and $n > k$ arbitrary.
	We show that $\dopt(\A_n,\lambda\A_k) \leq \eps$ where $\lambda \coloneqq \frac nk$.
	
	Let $\omega$ be a random map from $V(\A_n)=[n]$ to $V(\A_k)=[k]$ and we check that it is an overcast from $\A_n$ to $\lambda\A_k$.
	It is clear that $\omega$ maps positive tuples to positive tuples.
	Further, for $e \in E(\A_k)$ that is a simple edge (i.e.\ $e=(x,y)$ for some $x \neq y$),
	\[
	\Eg f^{\A_n}(g^{-1}(e))
	= \frac 1n {n \choose 2} \frac 2{k^2}
	= \frac{n-1}{k^2}
	= \lambda  \frac 1k \cdot (1-1/n).
	\]
	For $e \in E(\A_k)$ that is a loop (i.e.\ $e=(x,x)$ for some $x$),
	\[
	\Eg f^{\A_n}(g^{-1}(e))
	= \frac nk + \frac 1n {n \choose 2} \frac 1{k^2}
	\geq \lambda.
	\]
	As $n > \frac 2\eps$, therefore $\A_n \overcasts (1-\eps) \lambda \A_k$. By symmetry (as $k>\frac 2\eps$ also), $\lambda \A_k \overcasts (1-\eps) \A_n$.
	Thus, $\dopt(\A_n,\lambda\A_k) \leq \eps$ as claimed.
	It follows that $\cA$ is strongly tw-pliable.
\end{proof}

\end{document}